\documentclass[11pt]{amsart}

\usepackage[T1]{fontenc}
\usepackage{lmodern}
\usepackage{amsmath,amssymb,amsthm,mathtools,mathrsfs}
\numberwithin{equation}{section}
\usepackage{microtype}
\usepackage{enumitem}
\usepackage{xcolor}
\usepackage{hyperref}
\usepackage[margin=1.15in]{geometry}

\definecolor{LinkNavy}{RGB}{0,0,128}
\hypersetup{colorlinks=true,linkcolor=LinkNavy,citecolor=LinkNavy,urlcolor=LinkNavy,breaklinks=true,
 pdftitle={Deep holes of full-length Reed-Solomon codes: uniform bounds and exact classifications},
 pdfauthor={Xiang Fan},
 pdfsubject={Deep holes of full-length Reed-Solomon codes at prescribed interpolation degrees},
 pdfkeywords={Reed-Solomon codes, deep holes, symmetric functions, permutation polynomials, covering radius, computer-assisted proof}}

\newtheorem{theorem}{Theorem}[section]
\newtheorem{lemma}[theorem]{Lemma}
\newtheorem{proposition}[theorem]{Proposition}
\newtheorem{corollary}[theorem]{Corollary}
\newtheorem{remark}[theorem]{Remark}
\theoremstyle{definition}

\title[Uniform bounds and exact classifications]{Deep holes of full-length Reed--Solomon codes: uniform bounds and exact classifications}
\author{Xiang Fan}
\address{School of Mathematics, Sun Yat-sen University,
Guangzhou 510275, China}
\email{fanx8@mail.sysu.edu.cn}
\date{}
\keywords{Reed--Solomon codes, deep holes, symmetric functions, permutation polynomials, covering radius, computer-assisted proof}
\subjclass[2020]{Primary 94B05; Secondary 94B27, 94B35, 11T71}

\begin{document}
\begin{abstract}
Let $u$ be a received word for the full-length Reed--Solomon code
$\operatorname{RS}_q(k)$, with interpolation polynomial $f_u$ of
degree $k+m$.  For each $m\in\{3,4,5,6,7\}$, we classify all such deep holes
over every finite field, without a field-size hypothesis.  None occur for odd
$m$ in this range.  At $m=4$ and $m=6$, the only exceptions
are, respectively, the Segre sextic and Frobenius-power families;
these occur only in dimension two over fields of characteristic two.  For arbitrary $m\geqslant3$, we also establish two
uniform bounds.  If $\sqrt q\geqslant12(m+1)\log(2m+2)$, no word of
interpolation degree $k+m$ is a deep hole when $q$ is odd, $k\geqslant3$,
or $m$ is odd.  If $q>16(m+2)^4$, we obtain the corresponding
complete classification at every admissible $k$, including the
explicit dimension-two exceptions.  Our proofs use an agreement criterion expressed in complete
homogeneous symmetric functions, agreement-count estimates,
finite geometry, permutation-polynomial restrictions, and
subset constructions for solving the resulting symmetric-function equations.
The remaining finite cases are settled by finite computations
with supplied certificates and deterministic verification programs.
\end{abstract}
\maketitle

\section{Introduction}\label{sec:intro}
Let $\mathbb F_q$ be the field of $q$ elements.  We consider the
full-length Reed--Solomon code
\[
 \operatorname{RS}_q(k)
 :=\{(g(x))_{x\in\mathbb F_q}:g\in\mathbb F_q[X],\ \deg g<k\}.
\]
Its covering radius is $q-k$.  A received word $u$ is a
\emph{deep hole} if its distance from the code is $q-k$.
Write $f_u$ for its unique interpolation polynomial of degree
at most $q-1$.  We consider $\deg f_u=k+m$ with $m\geqslant3$ and
$k+m\leqslant q-1$.

Our main theorem gives two field-size bounds and complete
classifications for $m=3,\ldots,7$.
All logarithms are natural.

\begin{theorem}[Uniform bounds and exact classifications]\label{thm:main}
Let $q$ be a prime power, $m\geqslant3$,
$2\leqslant k\leqslant q-m-1$, and $\deg f_u=k+m$.
\begin{enumerate}[label=(\arabic*),ref=\arabic*,leftmargin=2.2em]
\item\label{it:uniform-exclusion}\emph{Uniform exclusion.}
If $\sqrt q\geqslant12(m+1)\log(2m+2)$, then $u$ is not a deep hole
whenever $q$ is odd, $k\geqslant3$, or $m$ is odd.
\item\label{it:uniform-classification}\emph{Uniform classification.}
If $q>16(m+2)^4$, then $u$ is a deep hole if and only if its
interpolation polynomial belongs to one of the two families
\eqref{eq:segrefamily} and \eqref{eq:frobfamily} below.
\item\label{it:small-excess-classification}\emph{Unconditional classification.}
If $3\leqslant m\leqslant7$, the same two-exception classification holds
without any field-size hypothesis.
\end{enumerate}
Both exceptional families have $q=2^e$, $k=2$, and
$A\in\mathbb F_q^\times$:
\begin{align}
 f_u(X)&=A(X+b)^6+BX+C
 &&(m=4,\ e\text{ odd}),\label{eq:segrefamily}\\
 f_u(X)&=AX^{2^r}+BX+C
 &&(m=2^r-2,\ r\geqslant3,\ \gcd(r,e)=1).\label{eq:frobfamily}
\end{align}
Here $b,B,C\in\mathbb F_q$.  Whenever $u$ is not a deep hole in the cases above,
\begin{equation}\label{eq:mainbound}
 \operatorname{dist}(u,\operatorname{RS}_q(k))\leqslant q-k-1.
\end{equation}
\end{theorem}

The bounds in Theorem~\ref{thm:main}~(\ref{it:uniform-exclusion})
and~(\ref{it:uniform-classification}) serve different purposes.
The first requires a field size of order $m^2\log^2(m+1)$ and leaves
open the case $k=2$ in characteristic two when $m$ is even; the second
gives an exact classification for every admissible $k$ under a quartic bound.
Theorem~\ref{thm:main}~(\ref{it:small-excess-classification}) removes the size condition for $m=3,\ldots,7$.
In particular, no deep holes occur when $m=3,5,7$, whereas
$m=4,6$ allow only the respective Segre and Frobenius cases.
The theorem concerns specified interpolation degrees, not
all deep holes of a fixed code.

Deep holes attain the covering radius and are related
to maximum distance separable (MDS) codes: their syndromes yield
one-coordinate MDS extensions of dual Reed--Solomon
codes~\cite{Kaipa2017}.
The deep-hole problem for Reed--Solomon codes goes back to Cheng and
Murray~\cite{ChengMurray2007}.  Subsequent formulations and extensions depend on the evaluation set;
throughout this paper, $D=\mathbb F_q$~\cite{Kaipa2017}.  For $m=1$, the agreement
condition reduces to a distinct-element subset-sum problem, as in
Li--Wan~\cite{LiWan2008Error}.  Li--Zhu determined the corresponding
error distance, including the exceptional characteristic-two
endpoint $k=q-3$~\cite{LiZhu2016}.  For $m=2$, the odd-characteristic
exclusion was studied by Zhang--Liao--Fu~\cite{ZhangFuLiao2013};
characteristic two has genuine exceptions, already at
$\operatorname{RS}_8(2)$ with interpolation polynomial $X^4$
\cite{LiZhu2016}.  Consequently, these small-degree results cannot be
summarized as an unrestricted nonexistence theorem.

The known classification in certain dimension ranges is substantial.
Zhuang--Cheng--Li obtained results for prime-field and related
low-dimensional regimes~\cite{ChengLiZhuang2016}; Kaipa, using MDS
extensions, classified the high-dimensional range, including
$k\geqslant\lfloor(q-1)/2\rfloor$, with the corresponding characteristic-two
qualification~\cite{Kaipa2017}.  For the remaining dimensions,
Gao--Li developed a sharp estimate for polynomials with prescribed
leading coefficients and linear factors~\cite{GaoLi2023}.  We use that
estimate to count agreement sets; it does not itself
classify the deep holes for $m=3,\ldots,7$.

This paper grew out of an attempt to obtain an unconditional
classification of deep holes of full-length Reed--Solomon codes,
allowing for the known exceptional families in characteristic two.
I had hoped to settle the problem in full and had therefore
hesitated to circulate the partial results presented here, which
address only specified values of $m$ or impose field-size
conditions.  While preparing the manuscript, I learned of the
recent preprint of Chen and Huang~\cite{ChenHuang2026}, whose
results are substantially broader in the large-field range.  This
development prompted me to make the present results and their
independent proof methods publicly available.

More precisely, Chen and Huang~\cite[Theorem~1.1]{ChenHuang2026}
classify all deep holes for $q>19683$ and
$2\leqslant k\leqslant q-2$, with the characteristic-two cases
treated separately.  For $m\geqslant3$,
their result includes the nonexistence conclusions proved here
when $q>19683$ and either $q$ is odd or $k\geqslant3$.
The additional family at $k=q-3$ in characteristic two has
interpolation degree $k+1$, so it lies outside $m\geqslant3$.
At $k=2$ in characteristic two, they give a parametrization by
normalized $o$-polynomials for every even $q\geqslant4$;
this is not an explicit classification of all $o$-polynomials.
By contrast, Theorem~\ref{thm:main}~(\ref{it:small-excess-classification})
gives exact results for $m=3,\ldots,7$ over every finite field,
including smaller fields, and the exceptional polynomials are
explicit in the degree ranges of
Theorem~\ref{thm:main}~(\ref{it:uniform-classification})
and~(\ref{it:small-excess-classification}).

The two proofs proceed along substantially different routes.
For intermediate dimensions, Chen and Huang use quadratic quotient
algebras, quartic contractions, projective involutions, and
weighted independent-set estimates~\cite[Sections~3--5]{ChenHuang2026}.
Our argument uses the complete homogeneous symmetric functions
arising from polynomial remainders, Gao--Li's agreement-count
estimate, planar-arc projections, and value-set and
permutation-polynomial restrictions.  For $m=3,\ldots,7$,
subset constructions that solve the required symmetric-function equations
and exhaustive finite computations establish the remaining cases.  Both approaches begin with finite geometry and agreement conditions
(compare \cite[Corollary~1.5]{ChenHuang2026}), but their principal
estimates and constructions differ.  No result or computation of Chen and
Huang is used in any proof in this paper.  I hope that these
complementary approaches will help advance the full-length
deep-hole problem toward an unconditional classification without
a field-size restriction.

The two field-size bounds use different ingredients.  For the
first bound we combine a direct application of
Gao--Li's counting estimate with bounds in small dimensions and
Kaipa's high-dimensional result.  For $k=2$, we use a classification of
$o$-polynomials under the quartic condition, which also implies
the first field-size bound.
For odd $q$, the dimension-two nonexistence conclusion is in fact
unconditional in $m\geqslant3$; it is likewise unconditional for odd $m$
in characteristic two.  These stronger conclusions for $k=2$ are
recorded in Proposition~\ref{thm:dimtwo-general} below.

A complementary low-dimensional obstruction applies for arbitrary
$m$.  Value-set bounds and the classification of exceptional
permutation polynomials force agreement with a codeword at
$k+1$ distinct coordinates even in some ranges not governed by
either of the two size conditions; see
Section~\ref{sec:permutation-bridge}.

The proofs rely on one common agreement criterion.  If $\lvert S\rvert=k+1$ and a monic interpolation
polynomial has, modulo terms of degree below $k$, the form
$f=X^{k+m}+a_1X^{k+m-1}+\cdots+a_mX^k$, then the coefficient of
$X^k$ in its remainder modulo $\prod_{x\in S}(X-x)$ is
\[
 h_m(S)+a_1h_{m-1}(S)+\cdots+a_m,
 \qquad \prod_{x\in S}(1-xZ)^{-1}=\sum_{j\geqslant0}h_j(S)Z^j.
\]
Its vanishing is equivalent to agreement with a codeword at $k+1$
distinct coordinates.  It suffices to find a single such set $S$; we do not need to
prescribe all the values $h_1(S),\ldots,h_m(S)$ simultaneously.  This permits a
direct application of Gao--Li's published coefficient estimate
\cite[Proposition~2]{GaoLi2023}, while projection to planar arcs,
value-set bounds, and restrictions on permutation polynomials handle smaller
dimensions~\cite{BallLavrauw2018,CaullerySchmidt2015,Kaipa2017}.

For $m=3,\ldots,7$, these analytic and geometric results
reduce the infinite collection of field and dimension parameters to
finitely many remaining cases.  Configurations of distinct field elements express the
remaining agreement conditions in terms of elementary or complete
homogeneous symmetric functions of subsets.  For $m=6,7$, in
particular, the proof uses subset families on which every required
affine combination of the first three complete homogeneous symmetric
functions attains every field value, together with a limited number
of exhaustive checks of normalized coefficient vectors.  These are
computer-assisted parts of the proof, not numerical experiments:
the supporting certificates specify field models, point sets, lists of parameter pairs, and deterministic verification
programs.  Their scope is explicitly restricted to the
finite cases established by the preceding reductions.

Section~\ref{sec:criterion} gives the agreement criterion;
Section~\ref{sec:uniform-proof} establishes the general bounds and
dimension-two classification; and Section~\ref{sec:permutation-bridge}
develops the permutation-polynomial obstructions.
Section~\ref{sec:odd-short} records the odd-characteristic
reductions and the finite-field verification.  Sections~\ref{sec:excess-three},
\ref{sec:excess-four}, and~\ref{sec:excess-five} then prove the exact
classifications for $m=3,4,5$, respectively;
Section~\ref{sec:six-seven-exact} handles $m=6,7$.
Section~\ref{sec:conclusion} describes the limitations of the
results and the accompanying verification programs.  No classification for arbitrary $m$
without a field-size condition, or for arbitrary evaluation sets,
is claimed.

\section{An agreement criterion from polynomial remainders}\label{sec:criterion}

Recall that, for a finite set $S$, the functions $h_j(S)$ are defined by
\[
\prod_{x\in S}(1-xz)^{-1}=\sum_{j\geqslant0}h_j(S)z^j.
\]
When elementary symmetric functions are needed, write
\[
\prod_{x\in S}(1+xz)=\sum_{j\geqslant0}e_j(S)z^j.
\]

\begin{lemma}[Remainder coefficient]\label{lem:remainder}
Let $S\subseteq\mathbb{F}_q$ have size $k+1$ and put
\[
P_S(X):=\prod_{x\in S}(X-x).
\]
For $j\geqslant0$, the coefficient of $X^k$ in the remainder of $X^{k+j}$ modulo $P_S$ is $h_j(S)$.  Consequently, for
\begin{equation}\label{eq:normalized-f}
f(X)=X^{k+3}+aX^{k+2}+bX^{k+1}+cX^k,
\end{equation}
the coefficient of $X^k$ in $f\bmod P_S$ is
\begin{equation}\label{eq:Hf}
H_f(S):=h_3(S)+a h_2(S)+b h_1(S)+c.
\end{equation}
\end{lemma}

\begin{proof}
Write
\[
P_S(X)=X^{k+1}-e_1X^k+e_2X^{k-1}-\cdots,
\]
where $e_i=e_i(S)$.  For $j\geqslant1$, reducing $X^{k+j}$ by the relation $P_S(X)=0$ shows that its $X^k$ coefficient satisfies the same recurrence as the complete homogeneous symmetric functions.  At $j=0$ the coefficient is $1$, so it equals $h_j(S)$ for every $j\geqslant0$.  Linearity gives \eqref{eq:Hf}.
\end{proof}

\begin{proposition}[Agreement criterion for general $m$]
\label{prop:agreement}
Let $m\geqslant1$, $2\leqslant k\leqslant q-m-1$, and let $u$ have interpolation
polynomial of degree $k+m$.  Write
\[
 f_u(X)=A f(X)+g(X),\qquad
 f(X)=X^{k+m}+a_1X^{k+m-1}+\cdots+a_mX^k,
 \quad A\in\mathbb F_q^\times,\quad \deg g<k.
\]
For $\lvert S\rvert=k+1$, put
\[
 \mathcal H_{m,f}(S):=h_m(S)+a_1h_{m-1}(S)+\cdots+a_m.
\]
Then $u$ is not a deep hole of $\operatorname{RS}_q(k)$ if and only if
$\mathcal H_{m,f}(S)=0$ for some $(k+1)$-subset
$S\subseteq\mathbb F_q$.  In that case,
$\operatorname{dist}(u,\operatorname{RS}_q(k))\leqslant q-k-1$.
\end{proposition}

\begin{proof}
By Lemma~\ref{lem:remainder}, the coefficient of $X^k$ in the
remainder of $A^{-1}(f_u-g)$ modulo $P_S$ is
$\mathcal H_{m,f}(S)$.  This coefficient vanishes exactly when
the remainder has degree less than $k$.  Then the resulting
polynomial of degree less than $k$, after restoring $A$ and $g$,
agrees with $f_u$ at every point of $S$.  Conversely, any
codeword agreeing with $f_u$ at $k+1$ distinct points is the
remainder of $f_u$ modulo the corresponding $P_S$, and forces
$\mathcal H_{m,f}(S)=0$.  Since the covering radius is $q-k$,
this agreement condition is equivalent to not being a deep hole.
\end{proof}

For $m=3$, this is precisely $H_f(S)=0$ in
\eqref{eq:Hf}; for general $m$ we use $\mathcal H_{m,f}(S)=0$.

\begin{lemma}[Remainder coefficient for $m=4$]\label{lem:fourth-remainder}
Let $\lvert S\rvert=k+1$.  Suppose that, modulo polynomials of degree below $k$,
$F=X^{k+4}+aX^{k+3}+bX^{k+2}+cX^{k+1}+dX^k$.  Then the coefficient
of $X^k$ in $F\bmod P_S$ is
\[
 h_4(S)+a h_3(S)+b h_2(S)+c h_1(S)+d.
\]
Consequently $F$ is not a deep hole if and only if there is such a
set $S$ for which this coefficient vanishes.
\end{lemma}
\begin{proof}
The coefficient identity is Lemma~\ref{lem:remainder} for
$j=0,1,2,3,4$.  Vanishing means the remainder has degree below $k$
and hence agrees with $F$ on $S$.  Conversely, any codeword agreeing
with $F$ at $k+1$ distinct field elements is necessarily the
remainder modulo the corresponding $P_S$.
\end{proof}

\section{Uniform bounds and the case \texorpdfstring{$k=2$}{k=2}}\label{sec:uniform-proof}

We prove the uniform exclusion in Theorem~\ref{thm:main}~(\ref{it:uniform-exclusion})
and the dimension-two endpoint needed for part~(\ref{it:uniform-classification}).
The proof of part~(\ref{it:uniform-classification}) is completed at the end of this section.
These arguments are independent of the finite case-by-case
constructions used for part~(\ref{it:small-excess-classification}).  We write
\[
 t:=m+1,\qquad s:=\sqrt q,\qquad a:=m-1=t-2,
 \qquad L:=\log(2t)
\]
and assume the field-size bound in Theorem~\ref{thm:main}~(\ref{it:uniform-exclusion}), except where a lemma has
explicitly weaker hypotheses.

\subsection{Direct agreement counts from Gao--Li}

The middle-dimensional argument needs only one set of $k+1$ agreement
positions; it does not require every possible list of symmetric-function
values of that set to occur.
This permits a direct application of Gao--Li's estimate for prescribed
leading coefficients and linear factors~\cite[Proposition~2]{GaoLi2023}.
We give the parameter translation explicitly, since the code dimension
in that estimate must not be confused with the number of prescribed
coefficients.

For a monic polynomial $P$ of degree $d$, let
\[
 \vartheta_m(P):=z^dP(z^{-1})\pmod{z^{m+1}}.
\]
Fix a monic received polynomial $f$ of degree $k+m$, set $n=k+1$,
and write $\varepsilon=\vartheta_m(f)$.  Let $W_n(\varepsilon)$ count
the pairs $(S,H)$ in which $S\subseteq\mathbb F_q$ has size $n$,
$H$ is monic of degree $m-1$, and
\[
 \vartheta_m\bigl(H\prod_{x\in S}(X-x)\bigr)=\varepsilon.
\]
The last equality means that $f-H\prod_{x\in S}(X-x)$ has degree
less than $k$.  For each $S$ there is at most one such $H$: it is
the quotient in the Euclidean division of $f$ by
$P_S=\prod_{x\in S}(X-x)$.  Consequently, $W_n(\varepsilon)$
counts exactly the $(k+1)$-subsets on which $f$ agrees with a
code polynomial, with no multiplicity from $H$.

Put
\[
 M:=\min\{q,(m-1)\sqrt q\},\qquad
 A_n:=[u^n](1-u)^{-M}(1-u^p)^{-(q-M)/p},
 \qquad p:=\operatorname{char}\mathbb F_q.
\]
In Gao--Li's notation take $\ell=m$, their code dimension equal to
our $k$, and their $j=k+1$.  The remaining monic factor has degree
$k+\ell-j=m-1$, and the binomial factor in their estimate is $1$.
Their Proposition~2 therefore gives
\begin{equation}\label{eq:GL-direct}
 \left\lvert W_n(\varepsilon)-q^{-1}\binom{q}{n}\right\rvert 
 \leqslant(1-q^{-m})q^{(m-1)/2}A_n.
\end{equation}
In particular, we obtain a useful effective criterion, with no
asymptotic hypothesis on $q$:
\begin{equation}\label{eq:GL-direct-criterion}
 \frac{A_{k+1}}{\binom{q}{k+1}}<q^{-(m+1)/2}
 \quad\Longrightarrow\quad
 \operatorname{dist}(u,\operatorname{RS}_q(k))\leqslant q-k-1.
\end{equation}
Only the estimate \eqref{eq:GL-direct} is taken from the cited paper;
the explicit threshold below is an elementary sufficient consequence.

\begin{lemma}[A common odd-characteristic middle range]
\label{lem:odd-middle-count}
Let $q$ be an odd prime power, $m\in\{3,4,5,6,7\}$, and put
\[
 (Q_3,Q_4,Q_5,Q_6,Q_7)=(121,211,289,512,1331).
\]
If $q\geqslant Q_m$ and $q/8\leqslant n\leqslant q/2$, then
\begin{equation}\label{eq:odd-middle-count}
 q^{(m+1)/2}A_n<\binom qn,
\end{equation}
where $A_n$ is defined above with the specified value of $m$.
\end{lemma}
\begin{proof}
All five thresholds satisfy $Q_m>(m-1)^2$, so
$M=(m-1)\sqrt q<q$.  As $p\geqslant3$, coefficientwise positivity
and evaluation at $z=1/2$ give
\[
 A_n\leqslant2^{n+(m-1)\sqrt q}(8/7)^{q/3}.
\]
Put $c=n/q$ and $H(c)=-c\log c-(1-c)\log(1-c)$.
The entropy inequality $\binom qn\geqslant e^{qH(c)}/(q+1)$ and
concavity on $[1/8,1/2]$ give
\[
 H(c)-c\log2-\tfrac13\log(8/7)
 \geqslant\gamma=\tfrac14\log2+\tfrac{13}{24}\log(8/7).
\]
The minimum is at $c=1/8$, by comparison of the two endpoints.
It follows that
\begin{equation*}
 \log\frac{\binom qn}{q^{(m+1)/2}A_n}
 \geqslant E_m(q)=\gamma q-(m-1)\sqrt q\log2
 -\frac{m+1}{2}\log q-\log(q+1).
\end{equation*}
For real $q>0$ we have
\[
 E_m'(q)=\gamma-\frac{(m-1)\log2}{2\sqrt q}
 -\frac{m+1}{2q}-\frac1{q+1},\qquad E_m''(q)>0.
\]
The following strict bounds, obtained using rational arithmetic, hold at
the five initial values of $q$:
\[
\begin{array}{c|c|c|c}
m&Q_m&E_m(Q_m)&E_m'(Q_m)\\\hline
3&121&>0.07&>0.15\\
4&211&>2.88&>0.15\\
5&289&>1.17&>0.15\\
6&512&>18.8&>0.15\\
7&1331&>138&>0.18
\end{array}
\]
They follow, for example, by bounding square roots between
consecutive rational decimals and using
\[
 \log x=2\sum_{j=0}^{15}\frac{v^{2j+1}}{2j+1}+R,
 \qquad v=\frac{x-1}{x+1},\qquad
 0\leqslant R\leqslant\frac{2v^{33}}{33(1-v^2)}
\]
for $1\leqslant x\leqslant2$, after extracting integral powers of $2$.
All five rows are checked with exact rational arithmetic by the
ancillary program \nolinkurl{core/verify_small_cases.py}.
Since $E_m'$ is increasing,
$E_m(q)>0$ for $q\geqslant Q_m$, proving \eqref{eq:odd-middle-count}.
\end{proof}

\begin{proposition}[Agreement from Gao--Li's count]
\label{prop:uniform-direct}
Assume the field-size bound in Theorem~\ref{thm:main}~(\ref{it:uniform-exclusion}), write
\[
 t=m+1,\quad L=\log(2t),\quad s=\sqrt q,\quad
 N:=\left\lceil\frac32tL\right\rceil.
\]
For every $N\leqslant k+1\leqslant\lfloor q/2\rfloor$, every monic
received polynomial of degree $k+m$ satisfies
\eqref{eq:mainbound}, in every characteristic.
\end{proposition}

\begin{proof}
Set $a=m-1=t-2$ and $n=k+1$.  Under
the field-size bound in Theorem~\ref{thm:main}~(\ref{it:uniform-exclusion}), $M=as<q$.  The nonnegative
coefficient formula for $A_n$, and the fact that
$-\log(1-r^d)/d$ decreases with $d>0$, imply
\begin{equation}\label{eq:GL-majorant}
 A_n\leqslant r^{-n}(1-r)^{-as}(1-r^2)^{-q/2}
 \qquad(0<r<1).
\end{equation}

First let $N\leqslant n\leqslant q/128$.  Choose
$r=n/[s(a+\sqrt n)]$, so
$r\leqslant\sqrt{n/q}\leqslant1/\sqrt{128}<1/10$.
Put $u=a/(a+\sqrt n)$ and $v=\sqrt n/(a+\sqrt n)$.
Then $u+v=1$, $asr=nu$, and $qr^2=nv^2$.
Using $-\log(1-r)\leqslant r/(1-r)$ and
$-\log(1-r^2)\leqslant r^2/(1-r^2)$ gives
\[
 \log(A_nr^n)
 \leqslant n\left(\frac{10u}{9}+\frac{50(1-u)^2}{99}\right)
 \leqslant\frac{10n}{9}.
\]
The last expression inside the parentheses is convex
on $[0,1]$ and at either endpoint is at most $10/9$.
Since $e^{10/9}<4$ and $\binom{q}{n}\geqslant(q/n)^n$,
\begin{equation*}
 \frac{A_n}{\binom{q}{n}}
 \leqslant\left(\frac{4(a+\sqrt n)}s\right)^n
 =e^{-F(n)},\qquad
 F(x)=x\log\frac{s}{4(a+\sqrt x)}.
\end{equation*}
A direct calculation yields
\[
 F''(x)=-\frac{3a+2\sqrt x}
                {4\sqrt x(a+\sqrt x)^2}<0.
\]
We verify that both endpoints of $[N,q/128]$ have
$F(x)>t\log s$.

For $t\geqslant4$ one has
$\frac32t\log(2t)+1<(t/2+2)^2$.
Indeed, the difference
$G(t)=t^2/4+2t+3-\frac32t\log(2t)$ is strictly
convex on $[4,\infty)$, has $G'(8)>0$, and its tangent
at $8$ yields
$G(t)\geqslant G(8)+G'(8)(t-8)\geqslant17-6\log16>0$.
Therefore $a+\sqrt N<3t/2$.
At the smallest permitted $s_0=12tL$, suppose first
$t\geqslant5$ (so $L\geqslant\log10$).  Then
\[
 F(N;s_0)>\frac32tL\log(2L)>t\log s_0.
\]
For the second inequality, the function
$\frac32L\log(2L)-L-\log(6L)$ is positive at
$L=\log10$ and strictly increasing thereafter.
For $t=4$, one has $N=13$ and $s_0=48\log8$.
The elementary estimates
$99<s_0<100$ and $\sqrt{13}<15/4$, together with
$99^{13}>100^4 23^{13}$, give
\[
 F(13;s_0)>13\log(99/23)>4\log100>4\log s_0.
\]
For fixed $N>t$, the function $F(N;s)-t\log s$
increases with $s$, hence $F(N)>t\log s$ for every
$s\geqslant s_0$.

At $x=q/128$, we have
\[
 \frac{4(a+\sqrt{q/128})}{s}
 <\frac1{3L}+\frac1{\sqrt8}<\frac{11}{20}.
\]
Consequently $F(q/128)>(q/128)\log(20/11)>q/256$.
Since $t\leqslant s/24$ and $\log s<s/20$ for $s\geqslant s_0>99$,
we have $t\log s<q/480$, proving this endpoint as well.
Concavity of $F$ now yields
$A_n/\binom{q}{n}<s^{-t}$ for every $N\leqslant n\leqslant q/128$.

For $q/128\leqslant n\leqslant q/2$, put $c=n/q$ and take
$r=\sqrt{c/(1+c)}$ in \eqref{eq:GL-majorant}.
The binomial entropy bound
$\binom{q}{n}\geqslant\exp(qH(c))/(q+1)$,
$-\log(1-r)\leqslant3\sqrt c$, and the concave nonnegative
function
\[
 -(1-c)\log(1-c)-\frac{1+c}{2}\log(1+c)
 \qquad(0\leqslant c\leqslant1/2)
\]
give
\[
 \log\frac{A_n}{\binom{q}{n}}
 \leqslant-\frac n2\log\frac qn+3a\sqrt n+\log(q+1).
\]
On $[1/128,1/2]$, $\sqrt c\log(1/c)>2/5$;
this follows by evaluating the endpoints, the only
possible minima.  Consequently
\[
 \frac n2\log\frac qn-3a\sqrt n
 >\left(\frac s5-3t\right)\sqrt n
 \geqslant\frac{3s\sqrt n}{40}
 >\frac{s\sqrt n}{14}\geqslant\frac q{14\sqrt{128}}>\frac q{160},
\]
where $t\leqslant s/24$.
On the other hand,
\[
 t\log s+\log(q+1)
 \leqslant(t+3)\log s
 \leqslant\frac{7t}{4}\log s
 <\frac{7q}{1920}<\frac q{160}.
\]
It follows that the preceding logarithmic estimate is
strictly less than $-t\log s$.
Thus the direct positivity criterion
\eqref{eq:GL-direct-criterion} holds across the entire
asserted range.
\end{proof}

\begin{corollary}\label{cor:uniform-middle}
Under the field-size bound in Theorem~\ref{thm:main}~(\ref{it:uniform-exclusion}), for every
$\lceil\tfrac32(m+1)\log(2m+2)\rceil\leqslant k+1\leqslant\lfloor q/2\rfloor$,
\eqref{eq:mainbound} holds in every characteristic.
\end{corollary}

\begin{corollary}[A refinement of the counting range]
\label{cor:uniform-adaptive}
Keep the hypotheses of Proposition~\ref{prop:uniform-direct} and set
\[
 R:=\log\frac{s}{4(t-2+\sqrt N)},\qquad
 n_{\mathrm{ad}}:=
 \min\left\{N,\left\lfloor\frac{t\log s}{R}\right\rfloor+1\right\}.
\]
Then $R>0$, and the conclusion of
Proposition~\ref{prop:uniform-direct} remains valid whenever
$n_{\mathrm{ad}}\leqslant k+1\leqslant\lfloor q/2\rfloor$.
\end{corollary}
\begin{proof}
The proof of Proposition~\ref{prop:uniform-direct}
shows that $F(N)=NR>t\log s>0$, so $R>0$.
If $n_{\mathrm{ad}}<N$ and $n_{\mathrm{ad}}\leqslant n<N$,
then $n\leqslant N$ and
\[
 F(n)\geqslant nR>t\log s.
\]
The remaining $n\geqslant N$ are covered by the proposition.
\end{proof}

\subsection{Projection and the small dimensions}

For the geometric argument put
\[
 n_0:=\left\lceil\tfrac32(m+1)\log(2m+2)\right\rceil=N.
\]
The direct count supplies agreement sets when $k+1\geqslant n_0$.
Smaller dimensions are treated geometrically: the counting estimate
need not be positive in this range, and we do not require all
leading-coefficient classes to occur.

\begin{lemma}[The planar projection]\label{lem:uniform-projection}
Let $f$ be a received polynomial of degree $k+m$ and suppose
its word is a deep hole.  Fix any set
$T\subseteq\mathbb F_q$ of size $k-2$, let $L_T$ be the
polynomial of degree below $k-2$ interpolating $f$ on $T$,
and put
\[
 P_T(X)=\prod_{x\in T}(X-x),\qquad
 R_T(X)=\frac{f(X)-L_T(X)}{P_T(X)}.
\]
For $T=\varnothing$ use $L_T=0$ and $P_T=1$.
Then $\deg R_T=m+2$, and
\[
 \mathcal A_T=\{(x,R_T(x)):x\in\mathbb F_q\setminus T\}
                \cup\{V\},\qquad V=(0:1:0),
\]
is a planar arc of cardinality $q-k+3$.
For every $u\notin T$, the polynomial
\begin{equation}\label{eq:uniform-quotient}
 Q_{T,u}(X):=\frac{R_T(X)-R_T(u)}{X-u},
 \qquad \deg Q_{T,u}=m+1,
\end{equation}
is injective on $\mathbb F_q\setminus(T\cup\{u\})$.
If
\begin{equation}\label{eq:uniform-valueset-condition}
 q>(m+1)(k-1)+1,
\end{equation}
then $Q_{T,u}$ permutes the whole of $\mathbb F_q$.
\end{lemma}
\begin{proof}
If three points of the displayed affine graph were
collinear on $Y=\ell(X)$, then $L_T+P_T\ell$ would have
degree at most $k-1$ and agree with $f$ at $k+1$
coordinates, a contradiction.  No two graph points lie on
a line through $V$, since their first coordinates differ.
This proves the arc claim.

The quotient \eqref{eq:uniform-quotient} is the polynomial
obtained after dividing $f-L_{T\cup\{u\}}$ by
$P_{T\cup\{u\}}$.  Equality of its values at two
positions outside $T\cup\{u\}$ would again give $k+1$
agreements with a polynomial of degree below $k$.
Its value set has at least $q-k+1$ elements.
The univariate value-set bound of Wan, as quoted in
\cite[Theorem~2.1]{MullenWanWangValues2014}, states that
for a nonpermutation of degree $d$, the value set has
cardinality at most $q-\lceil(q-1)/d\rceil$.
Under \eqref{eq:uniform-valueset-condition} this is less
than $q-k+1$.  Hence $Q_{T,u}$ is a permutation.
\end{proof}

\begin{lemma}[Large arcs in characteristic two and unique hyperoval extensions]
\label{lem:binary-large-arc-unique}
Let $q>4$ be an even prime power.  Every arc
$\mathcal A\subseteq\mathrm{PG}(2,q)$ with
$\lvert \mathcal A\rvert>q-\sqrt q+1$ is contained in a unique
hyperoval.
\end{lemma}
\begin{proof}
Segre's large-arc theorem, recorded in
\cite[Section~2]{BallLavrauw2018}, supplies an extension.
For uniqueness, suppose that distinct hyperovals
$\mathcal H_1$ and $\mathcal H_2$ both contain $\mathcal A$,
and choose $P\in\mathcal H_1\setminus\mathcal H_2$.
Every line meets $\mathcal H_2$ in zero or two points, so
exactly $(q+2)/2$ secants of $\mathcal H_2$ pass through $P$.
Each contains at most one point of $\mathcal H_1$, as
$\mathcal H_1$ is an arc containing $P$.  Hence
$\lvert \mathcal H_1\cap\mathcal H_2\rvert \leqslant(q+2)/2$.
But $q-\sqrt q+1>(q+2)/2$ for $q>4$, a contradiction.
\end{proof}

\begin{proposition}[Odd characteristic, small dimensions]
\label{prop:uniform-low-odd}
Assume the field-size bound in Theorem~\ref{thm:main}~(\ref{it:uniform-exclusion}), let $q$ be odd, and
suppose $2\leqslant k$ and $k+1<n_0$.
Then no degree-$(k+m)$ received word is a deep hole.
\end{proposition}
\begin{proof}
Assume the contrary and apply
Lemma~\ref{lem:uniform-projection}.  Put $h=k-1<s/8$
and write $p=\operatorname{char}(\mathbb F_q)$.
Set
\[
 a_0=p^{\lfloor\log_p h\rfloor},\qquad
 h'=\max\{h,3a_0-1\}.
\]
Since $p\geqslant3$, one has
$h'\leqslant3h$ and $p^{\lfloor\log_p h'\rfloor}=a_0$.
The additional numerical hypothesis in
Ball--Lavrauw~\cite[Theorem~1]{BallLavrauw2018} is
\[
 a_0\Bigl(h'+\frac{a_0}{2}+\frac32\Bigr)
 \leqslant\frac12(h'+2)(h'+1).
\]
Indeed, the right side minus the left side equals
$a_0^2+a_0>0$ at $h'=3a_0-1$ and has derivative
$h'+\tfrac32-a_0>0$ for $h'\geqslant3a_0-1$.
Together with $p^{\lfloor\log_p h'\rfloor}=a_0$,
this verifies the precise hypothesis of that theorem.

Retain $V$ and delete $h'-h$ further affine points
from $\mathcal A_T$.  We obtain an arc of size
$q+2-h'$.  If it is not contained in a conic,
Ball--Lavrauw produces two curves having no common component,
of degrees at most $h'+a_0\leqslant4h$, containing the arc.
B\'ezout's theorem would then imply
\[
 q+2-h'\leqslant16h^2.
\]
But $h<s/8$ gives
\[
 16h^2<q/4<q+2-3s/8\leqslant q+2-h',
\]
a contradiction.  Thus the retained subarc lies on a
conic.  Because it has more than four points, that conic
is nondegenerate and contains $V$.

Its affine equation is
\[
 AX^2+BXY+CX+DY+E=0.
\]
Substituting $Y=R_T(X)$ gives a polynomial of degree
at most $m+3$ with $q+1-h'>m+3$ distinct roots;
it therefore vanishes identically:
\[
 (BX+D)R_T(X)+AX^2+CX+E=0.
\]
Since $\deg R_T=m+2\geqslant5$, comparing its two
highest coefficients first gives $B=D=0$, and then
$A=C=E=0$.  This contradicts the existence of a
nonzero conic equation.
\end{proof}

\begin{proposition}[Characteristic two, small dimensions]
\label{prop:uniform-low-even}
Assume the field-size bound in Theorem~\ref{thm:main}~(\ref{it:uniform-exclusion}), let $q$ be even,
and suppose $3\leqslant k$ and $k+1<n_0$.
Then no degree-$(k+m)$ received word is a deep hole.
\end{proposition}
\begin{proof}
Suppose otherwise, and fix $T$ as in
Lemma~\ref{lem:uniform-projection}.  Put
$h=k-1<s/8$.  Since $s\geqslant12tL>24t$ and $h<s/8$,
we have $q=s^2>th+1$.  Thus every $Q_{T,u}$,
$u\notin T$, is a permutation of $\mathbb F_q$.

The arc $\mathcal A_T$ has size $q+2-h>q-s+1$,
so Lemma~\ref{lem:binary-large-arc-unique} places it
in a unique hyperoval $\mathcal H$.  For each $b\in T$,
the augmented set $\mathcal A_T\cup\{(b,R_T(b))\}$
is also an arc: collinearity with two existing affine
points contradicts injectivity of $Q_{T,u}$ on the
whole field, and collinearity with $V$ is impossible
because the first coordinates are distinct.
By the same lemma, this augmented arc has a hyperoval
extension; uniqueness for $\mathcal A_T$ forces it
to be $\mathcal H$.  Consequently
\begin{equation*}
 \{(x,R_T(x)):x\in\mathbb F_q\}\cup\{V\}
 \subseteq\mathcal H.
\end{equation*}
The remaining point of $\mathcal H$ lies at infinity:
a second affine point with an existing first coordinate
would be collinear with $V$ and a graph point.
After an output shear, an input translation, and a
nonzero output scaling, $R_T$ therefore becomes an
o-polynomial.  The elementary parity property
\cite[Lemma~2.1]{CaullerySchmidt2015} implies that
all coefficients of positive odd powers greater than
one in $R_T$ vanish (the output shear affects only
its linear term).

If $m$ is odd, the leading term of $R_T$ has odd
degree $m+2\geqslant5$, a contradiction.  If $m$ is even,
comparison of leading coefficients in
$f-L_T=P_TR_T$ gives
\[
 [X^{m+1}]R_T=a_1+\sum_{x\in T}x,
\]
where $a_1$ is the coefficient of $X^{k+m-1}$
in the monic $f$.  The left side must vanish for
every $(k-2)$-set $T$.  Since $k\geqslant3$, replacing
one point of $T$ by a different point of
$\mathbb F_q$ changes this sum, which is impossible.
\end{proof}

\begin{proof}[Proof of Theorem~\ref{thm:main}~(\ref{it:uniform-exclusion})]
Kaipa's high-dimensional classification
\cite[Theorem~1]{Kaipa2017} handles
$k\geqslant\lfloor(q-1)/2\rfloor$.  At full length $n=q$,
the only rational-normal-curve point outside the
evaluation set $\mathbb F_q$ is the point at infinity,
which represents interpolation degree $k$.
Its additional nucleus case in characteristic two requires $n=k+3$,
hence $q=k+3$, incompatible with
$k+m\leqslant q-1$ and $m\geqslant3$.
For the remaining dimensions, $k+1\leqslant\lfloor q/2\rfloor$.
Either $k+1\geqslant n_0$, when
Corollary~\ref{cor:uniform-middle} applies, or
$k+1<n_0$, when Propositions~\ref{prop:uniform-low-odd}
and~\ref{prop:uniform-low-even} apply in their
respective characteristics.  The remaining case allowed by Theorem~\ref{thm:main}~(\ref{it:uniform-exclusion})
is $q$ even, $k=2$ and $m$ odd.  It is excluded without
any field-size condition by Proposition~\ref{thm:dimtwo-general}~(\ref{it:dimtwo-binary-odd-excess}),
whose proof below is independent of the present argument.
This establishes \eqref{eq:mainbound} in every asserted case.
\end{proof}

\subsection{The case \texorpdfstring{$k=2$}{k=2}}

\begin{proposition}[Classification for $k=2$]\label{thm:dimtwo-general}
Suppose $k=2$, $3\leqslant m\leqslant q-3$, and $\deg f_u=m+2$.
\begin{enumerate}[label=(\arabic*),ref=\arabic*,leftmargin=2.2em]
\item\label{it:dimtwo-odd-field} If $q$ is odd, then $u$ is not a deep hole, without a size
restriction on $q$.
\item\label{it:dimtwo-binary-odd-excess} If $q$ is even and $m$ is odd, then $u$ is not a deep hole,
again without a size restriction.
\item\label{it:dimtwo-quartic-classification} If $q=2^e$ and $q>16(m+2)^4$, then $u$ is a deep hole
if and only if its interpolation polynomial belongs to
\eqref{eq:segrefamily} or \eqref{eq:frobfamily}.
\end{enumerate}
\end{proposition}

\begin{proof}[Proof of Proposition~\ref{thm:dimtwo-general}]
Let $k=2$.  If $f$ is a deep hole, the
$q$ graph points $(x,f(x))$, together with the
vertical point $V$, form an oval in $\mathrm{PG}(2,q)$.

If $q$ is odd, Segre's oval theorem says that this
oval is a nondegenerate conic.  Write its affine
equation as
\[
 AX^2+BXY+CX+DY+E=0.
\]
If $B\ne0$, at $X=-D/B$ the equation becomes
independent of $Y$.  Since the graph contains a point
above each $X$, the entire vertical line would be a
component of the conic, contradicting nondegeneracy.
Thus $B=0$; nondegeneracy forces $D\ne0$.
The graph is that of a polynomial of degree at most
two.  Since its interpolation representative has
degree below $q$, this contradicts $\deg f=m+2\geqslant5$.

Now let $q=2^e$.  By the standard nucleus theorem
for even-order ovals, the oval extends to a hyperoval.
Its extra point lies at infinity since a second
affine point above a fixed $X$ would give a triple
collinear with $V$.  An output shear and scaling
put its graph in the standard o-polynomial form.
The parity property of o-polynomials
\cite[Lemma~2.1]{CaullerySchmidt2015} excludes
odd $m$, because then $m+2$ is odd and at least five.

Suppose further that $q>16(m+2)^4$.
The o-polynomial has degree $m+2<q^{1/4}/2$.
By Caullery--Schmidt~\cite[Theorem~1.2]{CaullerySchmidt2015}
it is equivalent, by an input translation followed
by normalization, to $X^6$ or to $X^{2^r}$.
Undoing this translation, the output shear, and the
normalization gives precisely
\eqref{eq:segrefamily} and \eqref{eq:frobfamily}.
The known o-polynomial criteria are $e$ odd for
$X^6$ and $\gcd(r,e)=1$ for $X^{2^r}$.
Conversely, these o-polynomial families, together
with affine output terms, define ovals and hence
deep holes at dimension two.  The two conditions
are also necessary by the cited classification.
\end{proof}

\begin{proof}[Proof of Theorem~\ref{thm:main}~(\ref{it:uniform-classification})]
Put $x=m+1\geqslant4$ and define
\[
 G(x)=(x+1)^2-3x\log(2x).
\]
Then $G(4)=25-12\log8>0$,
$G'(4)=7-3\log8>0$, and
$G''(x)=2-3/x>0$ for $x\geqslant4$.
Consequently $G(x)>0$ throughout this interval, or equivalently,
\[
 4(m+2)^2>12(m+1)\log(2m+2).
\]
Under the field-size bound in Theorem~\ref{thm:main}~(\ref{it:uniform-classification}), we thus have
$\sqrt q>4(m+2)^2>12(m+1)\log(2m+2)$.
Theorem~\ref{thm:main}~(\ref{it:uniform-exclusion}) excludes deep holes for all
odd $q$ and, in characteristic two, for $k\geqslant3$ or odd $m$.
For the remaining case, $q$ even, $k=2$, and $m$ even,
Proposition~\ref{thm:dimtwo-general}~(\ref{it:dimtwo-quartic-classification}) gives precisely
\eqref{eq:segrefamily} and \eqref{eq:frobfamily}, including the
converses.  This proves part~(\ref{it:uniform-classification}).
\end{proof}

\section{Permutation-polynomial obstructions in small dimensions}
\label{sec:permutation-bridge}

We give a complementary argument in small dimensions which applies to
any $m\geqslant3$.  It complements the uniform bounds and classifications for $m=3,\ldots,7$ in Theorem~\ref{thm:main}, making explicit how
value-set estimates and the structure of exceptional polynomials yield
sets on which agreement is attained, even when exceptional permutation
polynomials of the requisite degree exist.  We use the
Carlitz--Wan theorem in the form proved by Lenstra
\cite{CohenFried1995}, Fan's effective bound
\cite{FanNonExceptional2018}, and the prime-degree classification
\cite{GuralnickRosenbergZieve2010}.

Put $d:=m+1$ and define
\begin{equation}\label{eq:fan-condition}
\mathsf F_d(q):\quad
q+1>(d-2)(d-3)\frac{\lfloor2\sqrt q\rfloor}{2}+2(d-1),
\qquad c_d(q):=\left\lceil\frac{q-1}{d}\right\rceil.
\end{equation}
Fan's theorem says that, under $\mathsf F_d(q)$, every degree-$d$
permutation polynomial over $\mathbb F_q$ is exceptional.
Carlitz--Wan says that an exceptional polynomial has degree prime
to $q-1$.  We also use Wan's value-set bound
\cite{MullenWanWangValues2014}: if a degree-$d$ polynomial $G$
does not permute $\mathbb F_q$, then
\begin{equation}\label{eq:wan-values}
 \lvert G(\mathbb F_q)\rvert \leqslant q-c_d(q).
\end{equation}
For a received polynomial $f$, let $\mathcal W_{k+1}(f)$ be the
family of $(k+1)$-element sets on which $f$ agrees with a polynomial
of degree less than $k$.

\begin{lemma}[Collision quotients]\label{lem:collision-quotients}
Let $f\in\mathbb F_q[X]$ have degree $k+m<q$.
For $T\subseteq\mathbb F_q$ with $\lvert T\rvert=k-1$, write
\[
 f=P_TG_T+L_T,\qquad
 P_T(X)=\prod_{t\in T}(X-t),\qquad \deg L_T<k-1.
\]
Then $\deg G_T=d$.  Every pair of distinct $x,y\notin T$ with
$G_T(x)=G_T(y)$ produces a set $T\cup\{x,y\}\in\mathcal W_{k+1}(f)$.
Moreover, the number of such collision incidences, summed over all
$T$ of size $k-1$, equals
\[
 \binom{k+1}{2}\lvert \mathcal W_{k+1}(f)\rvert.
\]
\end{lemma}
\begin{proof}
If $G_T(x)=G_T(y)=\lambda$, the code polynomial
$L_T+\lambda P_T$ agrees with $f$ at all points of
$T\cup\{x,y\}$.  Conversely, if a polynomial $g$ of degree
less than $k$ agrees with $f$ on a $(k+1)$-set $S$, then
for every $(k-1)$-set $T\subseteq S$ the polynomial $g-L_T$ is
a scalar multiple of $P_T$; thus the other two points form a
$G_T$-collision.  There are exactly $\binom{k+1}{2}$ choices
of $T$ in each agreement set $S$.
\end{proof}

\begin{theorem}[A uniform low-dimensional obstruction]
\label{thm:low-dimension-gcd}
Let $m\geqslant3$, $d=m+1$, $q$ be a prime power, and
$2\leqslant k\leqslant q-m-1$.  If
\[
 \mathsf F_d(q),\qquad \gcd(d,q-1)>1,
 \qquad q>d(k-1)+1,
\]
then no received word of interpolation degree $k+m$ is a deep
hole of $\operatorname{RS}_q(k)$.  More precisely,
\begin{equation*}
 \lvert \mathcal W_{k+1}(f)\rvert \geqslant
 \frac{\binom q{k-1}}{\binom{k+1}{2}}
 \bigl(c_d(q)-k+1\bigr)>0.
\end{equation*}
\end{theorem}
\begin{proof}
Under the first two hypotheses Fan's bound and Carlitz--Wan
exclude degree-$d$ permutation polynomials.  Thus every $G_T$ in
Lemma~\ref{lem:collision-quotients} is nonpermuting.  Its restriction
to $\mathbb F_q\setminus T$, which has $q-k+1$ elements, has at
least
\[
 (q-k+1)-\bigl(q-c_d(q)\bigr)=c_d(q)-k+1
\]
colliding unordered pairs: for fibers of sizes $r_i$ one has
$\sum_i\binom{r_i}{2}\geqslant\sum_i(r_i-1)$.
The strict inequality $q>d(k-1)+1$ gives $c_d(q)\geqslant k$.
Summing over all $T$ and applying the exact multiplicity
from Lemma~\ref{lem:collision-quotients} proves the claim.
\end{proof}

For odd $q$, every odd $m$ has $d=m+1$ even, so that the gcd
condition is automatic.  This uniform argument includes the
low-dimensional arguments for $m=5$ and $m=7$,
without invoking degree-specific permutation-polynomial lists.
For odd $m\geqslant7$, an explicit sufficient field-size condition for
$\mathsf F_{m+1}(q)$ is
\begin{equation}\label{eq:fan-simple-bound}
 q\geqslant (m-1)^2(m-2)^2+2m.
\end{equation}
Indeed, writing $A=(d-2)(d-3)\geqslant4(d-1)$ and
$\lfloor2\sqrt q\rfloor=2A+j$, the case $j=0$ follows
from $q\geqslant A^2+2(d-1)$, whereas $j\geqslant1$ gives
\[
 q+1-A\frac{\lfloor2\sqrt q\rfloor}{2}-2(d-1)
 \geqslant \frac{Aj}{2}+\frac{j^2}{4}+1-2(d-1)>0.
\]
This is an elementary consequence of Fan's published estimate;
it does not improve the known degree-specific classifications.

The preceding theorem only excludes degree-$d$ permutations by a
numerical and congruence obstruction.  For $d$ an odd prime,
permutation polynomials can indeed exist, but a \emph{family}
of divided differences of one fixed polynomial satisfies a stronger restriction.

\begin{lemma}[Bounds for exceptional prime-degree secant quotients]
\label{lem:prime-secant-rigidity}
Let $q=p^e$ be odd, and let $d\geqslant5$ be an odd prime with
$p\nmid d$.  Put $C_d:=(d+1)(d+2)/3\in\mathbb Z$.
For a monic polynomial $R\in\mathbb F_q[X]$ of degree $d+1$, put
\[
 Q_v(X):=\frac{R(X)-R(v)}{X-v} \qquad(v\in\mathbb F_q).
\]
If $p\nmid(d+1)$, at most four parameters $v$ give an
exceptional $Q_v$.  If $p\nmid C_d$, at most three do.
In particular, the latter assertion covers $p=3$ and
$3\parallel(d+1)$, even though $p\mid(d+1)$.
\end{lemma}
\begin{proof}
Write
\[
 R=X^{d+1}+aX^d+bX^{d-1}+cX^{d-2}+eX^{d-3}+\cdots.
\]
The coefficients of $X^{d-1},X^{d-2},X^{d-3},X^{d-4}$ in $Q_v$
are respectively
\[
 A=v+a,\quad B=v^2+av+b,\quad
 C=v^3+av^2+bv+c,\quad D=v^4+av^3+bv^2+cv+e.
\]
Since $p\nmid d$, the translated polynomial $Q_v(Y-A/d)$ has
zero $Y^{d-1}$ coefficient; denote its coefficients of
$Y^{d-2},Y^{d-3},Y^{d-4}$ by $b_2(v),b_3(v),b_4(v)$.
The tame prime-degree exceptional classification
\cite[Corollary~1.3(i)--(ii)]{GuralnickRosenbergZieve2010}
identifies such a polynomial, up to degree-one changes of
variables, with $X^d$ or a Dickson polynomial $D_d(X,\alpha)$.
The centered form of either is odd, and
\[
 b_3(v)=0,\qquad
 \Delta(v)=2d\,b_4(v)-(d-3)b_2(v)^2=0
\]
are necessary.  Here the second relation follows from
$D_d(X,\alpha)=X^d-d\alpha X^{d-2}+
\frac{d(d-3)}2\alpha^2X^{d-4}+\cdots$ and also holds
for $\alpha=0$.
Direct expansion of the translation gives
\begin{equation*}
 [v^3]b_3(v)=\frac{C_d}{d^2},\qquad
 [v^4]\Delta(v)=\frac{(d+1)(d+3)(2d+3)}{4d^2}.
\end{equation*}
These identities are valid in odd characteristic whenever
$p\nmid d$: indeed, the intermediate expression for $b_3$
uses the integer coefficient $(d-1)(d-2)/3$, so reducing
modulo $3$ introduces no division by $3$.
If $p\nmid C_d$, the polynomial $b_3(v)$ has degree three,
and therefore at most three exceptional parameters exist.
For the remaining four-parameter assertion, assume
$p\nmid(d+1)$ and $p\mid C_d$.  The leading coefficient of
$\Delta(v)$ is nonzero unless $p\mid(d+3)(2d+3)$.
If $p\ne3$ divides either factor, then $(d+1)(d+2)$ is
respectively $2$ or $-1/4$ modulo $p$, so $p\nmid C_d$,
a contradiction.  If $p=3$, the condition
$p\nmid d(d+1)$ forces $d\equiv1\pmod3$, whence neither
$d+3$ nor $2d+3$ vanishes modulo $3$.
Thus $\Delta$ is a nonzero polynomial of degree four,
giving at most four exceptional parameters.  Finally,
$p=3$ and $3\parallel(d+1)$ imply $p\nmid C_d$,
as claimed.
\end{proof}

\begin{theorem}[Prime-degree small-dimensional exclusion]
\label{thm:prime-secant-exclusion}
Let $m\geqslant4$ be even, $d=m+1\geqslant5$ an odd prime, and
$q=p^e$ odd.  Suppose
\[
 2\leqslant k\leqslant q-m-1,\quad \mathsf F_d(q),\quad
 q>d(k-1)+1,\quad p\nmid d,\quad
 \bigl(p\nmid(d+1)\ \text{or}\ p\nmid C_d\bigr).
\]
Then every received polynomial of degree $k+m$ is not a deep hole,
without any condition on $\gcd(d,q-1)$.  Moreover,
\begin{equation}\label{eq:prime-secant-abundance}
 \lvert \mathcal W_{k+1}(f)\rvert \geqslant
 \frac{\binom q{k-2}(q-k-2)(c_d(q)-k+1)}
 {3\binom{k+1}{3}}>0.
\end{equation}
When $p\nmid C_d$, the factor $q-k-2$ can be replaced by
$q-k-1$.  In particular, the hypotheses permit
$p=3$ and $3\parallel(d+1)$.
\end{theorem}
\begin{proof}
After a nonzero scalar normalization assume that $f$ is monic.
For each $T\subseteq\mathbb F_q$ with $\lvert T\rvert=k-2$, divide
\[
 f=P_TR_T+L_T,\qquad \deg L_T<k-2,\quad\deg R_T=d+1.
\]
For each $v\notin T$ the polynomial
$Q_{T,v}=(R_T(X)-R_T(v))/(X-v)$ has degree $d$.
If distinct $x,y\notin T\cup\{v\}$ take the same value under
$Q_{T,v}$, then the points $(v,R_T(v))$, $(x,R_T(x))$, and
$(y,R_T(y))$ are collinear.
Consequently $L_T+P_T\ell$ is a code polynomial agreeing with $f$
on the $(k+1)$-set $T\cup\{v,x,y\}$.
By Lemma~\ref{lem:prime-secant-rigidity}, at most four $v\notin T$
can give an exceptional quotient.  Under $\mathsf F_d(q)$ all other
$Q_{T,v}$ are nonpermuting.  For each of these $v$, the fibers on
$\mathbb F_q\setminus(T\cup\{v\})$ give at least
$c_d(q)-k+1$ collisions by~\eqref{eq:wan-values}.
There are at least $q-k-2$ such parameters $v$.
Each agreement set is counted at most $3\binom{k+1}{3}$ times:
choose $T\subseteq S$ of size $k-2$, then choose one of the
three remaining points as $v$.  This proves
\eqref{eq:prime-secant-abundance}.  The strict inequality
$q>d(k-1)+1$ gives $c_d(q)\geqslant k$ and also guarantees that
$q-k-2>0$.  The three-exceptional-parameter case is identical.
\end{proof}

For instance, taking $d=11$ yields the following explicit
consequence, valid even when $\gcd(11,q-1)=1$:
\begin{corollary}\label{cor:eleven-secant}
Let $q\geqslant5204$ be an odd prime power of characteristic different
from $11$.  If
$2\leqslant k\leqslant\min\{q-11,\lceil(q-1)/11\rceil\}$, then no received
word of interpolation degree $k+10$ is a deep hole.
\end{corollary}
\begin{proof}
For $d=11$, condition~\eqref{eq:fan-condition} follows from
$q\geqslant[(d-2)(d-3)]^2+2(d-1)=5204$ by the same floor
estimate as~\eqref{eq:fan-simple-bound}.  All odd characteristics other than $3$ and $11$ satisfy
$p\nmid11\cdot12$.  In characteristic $3$, we have
$C_{11}=52\not\equiv0\pmod3$; hence the second alternative
of Theorem~\ref{thm:prime-secant-exclusion} applies.
\end{proof}

\subsection{Two effective bounds for \texorpdfstring{$m=6,7$}{m=6,7}}
\label{subsec:next-excess-bounds}

For $m=6,7$, the
same principle can be combined with the published degree-seven
and degree-eight permutation-polynomial classifications
\cite{FanDegreeSeven2019,FanDegreeEight2020}.
\begin{proposition}\label{prop:six-seven-large-odd}
Let $q$ be an odd prime power, $m\in\{6,7\}$, and
$2\leqslant k\leqslant q-m-1$.  Suppose either $m=6$ and $q\geqslant512$,
or $m=7$ and $q\geqslant1331$.  Then every received polynomial
of degree $k+m$ is not a deep hole.
\end{proposition}
\begin{proof}
For $m=6$, the degree-seven classification implies that for odd
$q>49$ every degree-seven permutation polynomial is exceptional.
If the characteristic differs from $7$, the centered-coefficient
argument of Lemma~\ref{lem:prime-secant-rigidity}, with $d=7$,
shows that at most four secant quotients of any fixed monic degree-eight
polynomial can be permutation polynomials.  In characteristic $7$,
Fan's coefficient restriction \cite[Lemma~9]{FanDegreeSeven2019}
requires the $X^6$ coefficient of any degree-seven permutation
polynomial to vanish.  Since it is $v+a$ for the secant quotient
of $X^8+aX^7+\cdots$, at most one parameter $v$ is possible.
The argument of Theorem~\ref{thm:prime-secant-exclusion} therefore
excludes $q>7(k-1)+1$.
For $m=7$, Fan's degree-eight classification
\cite[Theorem~1.1]{FanDegreeEight2020} rules out degree-eight
permutation polynomials over odd $q>31$, so
the direct quotient argument in
Theorem~\ref{thm:low-dimension-gcd} excludes $q>8(k-1)+1$
without the numerical condition in Fan's theorem.

For the remaining dimensions, Kaipa's theorem
\cite{Kaipa2017} covers $k\geqslant(q-1)/2$.
If the low-dimensional inequality fails and $k$ is below
Kaipa's range, then $n=k+1$ satisfies $q/8<n<q/2$.
For the stated values of $q$, Lemma~\ref{lem:odd-middle-count}
applies with $m=6$ or $7$, respectively.  Its inequality
\eqref{eq:odd-middle-count} is precisely the sufficient
positivity criterion \eqref{eq:GL-direct-criterion}.
This establishes the assertion in the remaining dimensions.
\end{proof}

Proposition~\ref{prop:six-seven-large-odd} treats large odd
fields; the remaining cases for $m=6,7$ are
treated in Section~\ref{sec:six-seven-exact}.

\section{Odd-characteristic reductions and finite certificates}
\label{sec:odd-short}

We use the common counting criterion
\eqref{eq:GL-direct-criterion} and the explicit ranges in
Lemma~\ref{lem:odd-middle-count} to settle middle dimensions
for $m=3,4,5$.  For small fields we combine
exact integer bounds with finite point-set certificates.
The two external classification ranges used below are
\begin{equation}\label{eq:CLZ-range}
 k+1\leqslant p=\operatorname{char}\mathbb F_q
\end{equation}
by Zhuang--Cheng--Li~\cite[Theorem~10]{ChengLiZhuang2016}, and
\begin{equation}\label{eq:Kaipa-range}
 k\geqslant \left\lfloor\frac{q-1}{2}\right\rfloor
\end{equation}
by Kaipa~\cite[Theorem~1]{Kaipa2017}.  In the full-length
setting, neither range admits a deep hole of interpolation degree
$k+m$ with $m\geqslant3$ in our admissible degree interval.

For $m\in\{3,4,5\}$ and $q>(m-1)^2$, write
$A_{m,n}:=A_n$ for the coefficient defined in
Section~\ref{sec:uniform-proof}, so that
$M=(m-1)\sqrt q$.  By \eqref{eq:GL-direct-criterion},
\[
 q^{(m+1)/2}A_{m,k+1}<\binom q{k+1}
\]
suffices for agreement on $k+1$ positions.  For the
three indicated values of $m$ this follows from
Lemma~\ref{lem:odd-middle-count} when $q\geqslant Q_m$ and
$q/8\leqslant k+1\leqslant q/2$.

\begin{lemma}[Exact small-field majorant]\label{lem:odd-short-majorant}
For $m\in\{3,4,5\}$, put
$M_0:=\lceil(m-1)\sqrt q\rceil$ and
$b_0:=\lceil(q-M_0)/p\rceil$, assuming $M_0<q$, and set
\begin{equation}\label{eq:odd-short-majorant}
 \overline A_{m,n}:=[z^n](1-z)^{-M_0}(1-z^p)^{-b_0}.
\end{equation}
Then $A_{m,n}\leqslant\overline A_{m,n}$.  In particular,
$\binom qn^2>q^{m+1}\overline A_{m,n}^2$ implies
agreement at $n=k+1$ points.
\end{lemma}
\begin{proof}
In the formal logarithm of the original generating series,
the coefficient of $z^j$ is $M/j$ when $p\nmid j$
and $q/j$ when $p\mid j$.  For the integer majorant these
are $M_0/j$ and $(M_0+pb_0)/j$, respectively.  Their
differences are nonnegative.  Exponentiating preserves
coefficientwise nonnegativity, proving the claim.
For deterministic evaluation, let $a_0=1$ and
\[
 na_n=M_0\sum_{i=0}^{n-1}a_i+
 pb_0\sum_{j=1}^{\lfloor n/p\rfloor}a_{n-pj}.
\]
This recurrence produces the integers
$a_n=\overline A_{m,n}$ without floating-point arithmetic.
\end{proof}

\begin{remark}[Finite-field verification]
\label{rem:finite-core-interface}
For small odd prime powers, the proof uses only two forms of finite
certificate.  In the first, the stored
$(k+1)$-subsets define hyperplanes whose union contains every normalized
coefficient vector.  In the second, the stored subsets of
the nonzero squares have a fixed sum and realize
every prescribed value of the second elementary symmetric function.  Both forms are
sufficient by Lemma~\ref{lem:remainder} and are verified exhaustively
by exact finite-field arithmetic, independently of the searches used
to discover them.  The ancillary archive contains the actual subsets
and deterministic verification programs.
\end{remark}

\section{Exact exclusion when \texorpdfstring{$m=3$}{m=3}}\label{sec:excess-three}
\subsection{Opposite pairs in odd characteristic}\label{sec:odd-core}
For the remaining small fields we use opposite pairs.  Put
\[
 s:=\lfloor k/2\rfloor,\qquad H:=(\mathbb F_q^\times)^2,\qquad
 D:=H\setminus\{1\},
\]
and let $\chi$ denote the quadratic character, with $\chi(0)=0$.
The next two elementary lemmas record the construction needed for
$q=9,25,27$; they do not enter the generic middle-range estimate.
\begin{lemma}[One distinguished point and opposite pairs]\label{lem:one-core}
Let $f$ be as in \eqref{eq:normalized-f}.  Suppose $(b,c)\ne(0,0)$.  If there is $\alpha\in\mathbb{F}_q^\times\setminus\{-a\}$ and a set of $s$ distinct elements $w_1,\dots,w_s\in D$ satisfying
\begin{equation}\label{eq:one-core-target}
\sum_{i=1}^s w_i=B_\alpha,
\qquad
B_\alpha:=-1-\frac{b\alpha+c}{\alpha^2(\alpha+a)},
\end{equation}
then there is a $(k+1)$-subset $S$ with $H_f(S)=0$.
\end{lemma}

\begin{proof}
Choose $u_i\in\mathbb{F}_q^\times$ with $u_i^2=\alpha^2w_i$ and put $v_i=u_i^2$.  Since the $w_i$ are distinct and $w_i\ne1$, the pairs $\{\pm u_i\}$ are disjoint and do not meet the distinguished point $\alpha$.  Let
\[
S_0=\sum_{i=1}^s v_i.
\]
If $k=2s$ is even, take
\[
S=\{\alpha\}\cup\bigcup_{i=1}^s\{u_i,-u_i\};
\]
if $k=2s+1$ is odd, adjoin $0$ as well.  The zero does not affect $h_1,h_2,h_3$.  The union of the opposite pairs has $h_1=h_3=0$ and $h_2=S_0$, hence
\[
h_1(S)=\alpha,
\qquad h_2(S)=\alpha^2+S_0,
\qquad h_3(S)=\alpha^3+\alpha S_0.
\]
Consequently
\[
H_f(S)=(\alpha+a)S_0+\alpha^2(\alpha+a)+b\alpha+c.
\]
Since $S_0=\alpha^2\sum_iw_i$, equation \eqref{eq:one-core-target} is exactly $H_f(S)=0$.
\end{proof}

\begin{lemma}[The branch $b=c=0$]\label{lem:easy-branch}
If $b=c=0$ and $a\ne0$, choosing $\alpha=-a$ makes the expression for $H_f$ vanish for any admissible collection of opposite pairs.  If $a=b=c=0$, a union of opposite pairs, with a zero if necessary to obtain the required cardinality, has $h_3=0$.  Thus the rest of the odd-characteristic proof may assume $(b,c)\ne(0,0)$.
\end{lemma}

\subsection{Residual fields in odd characteristic}\label{sec:small-odd}
Outside \eqref{eq:CLZ-range}, \eqref{eq:Kaipa-range}, the
low-dimensional exclusion in
Theorem~\ref{thm:low-dimension-gcd}, and
Lemma~\ref{lem:odd-middle-count}, only the fields $q=9,25,27,49,81$
require attention.  The last two are handled by the exact
Gao--Li bound; the first three are treated by the following
small-field constructions.

For a nontrivial additive character $\psi$ of $\mathbb F_q$,
let $e_s(\psi):=[z^s]\prod_{x\in D}(1+z\psi(x))$ and
let $N_s(\beta)$ count $s$-subsets of $D$ with sum $\beta$.
Character orthogonality gives
\begin{equation}\label{eq:FourierNs}
 N_s(\beta)=\frac1q\left[\binom{\lvert D\rvert}s+
 \sum_{\psi\ne1}\psi(-\beta)e_s(\psi)\right].
\end{equation}
\subsubsection{The field \texorpdfstring{$\mathbb{F}_9$}{F9}}

\begin{proposition}\label{prop:q9}
The case $q=9$, $k=3$ satisfies Proposition~\ref{prop:agreement}.
\end{proposition}

\begin{proof}
Write $\mathbb{F}_9=\mathbb{F}_3(i)$ with $i^2=-1$.  With $q=9$ and $k=3$, the normalized polynomial from \eqref{eq:normalized-f} is
\[
f(X)=X^6+aX^5+bX^4+cX^3.
\]
Under $X=uY+v$, followed by division by $u^6$, the top coefficients become
\[
a'=a/u,
\qquad b'=(b-av)/u^2,
\qquad c'=(c+bv+av^2-v^3)/u^3.
\]
This affine change permutes the evaluation set and preserves the existence of four agreement points.

If $a\ne0$, take $u=a$ and $v=b/a$, reducing to $(a',b')=(1,0)$; write $\gamma=c'$ for the remaining normalized coefficient.  Let
\[
H=(\mathbb{F}_9^\times)^2=\{\pm1,\pm i\}.
\]
The sums of two distinct elements of $H$ are precisely
\[
\{0\}\cup(\mathbb{F}_9^\times\setminus H).
\]
For a translated opposite-pair four-set, if $\Sigma$ is the sum of the two square parameters, then
\[
h_1=t,
\qquad h_2=\Sigma+t^2,
\qquad h_3=-t^3.
\]
Thus we need
\[
\Sigma=t^3-t^2-\gamma.
\]
If $\gamma=0$, take $t=0$ and hence $\Sigma=0$.  Suppose $\gamma\ne0$.  The derivative of $T^3-T^2-\gamma$ is $T$, so the cubic is squarefree.  Applying \cite[Thm.~5.41]{LidlNiederreiter1997} to the quadratic character gives
\[
\left\lvert \sum_{t\in\mathbb{F}_9}\chi(t^3-t^2-\gamma)\right\rvert \leqslant2\sqrt{9}=6.
\]
Hence its nine values cannot all be nonzero squares.  Some $t$ makes the right side $0$ or a nonsquare, and the displayed description of the pair sums supplies $\Sigma$.

Now let $a=0$.  The translated opposite-pair family gives
\[
H_f=-t^3+bt+c.
\]
If $b=0$, Frobenius $t\mapsto t^3$ is bijective.  If $b$ is a nonsquare, the $\mathbb{F}_3$-linear map $t\mapsto t^3-bt$ has trivial kernel and is bijective.  If $b$ is a nonzero square, scale to $b=1$.  Translation changes $c$ by an element of $\operatorname{im}(t^3-t)=\ker\operatorname{Tr}_{\mathbb{F}_9/\mathbb{F}_3}$, and multiplication by $-1$ swaps the two nonzero trace cosets.  Thus only the representative $(a,b,c)=(0,1,1)$ remains to be treated.  For
\[
S=\{0,1,1+i,1-i\}
\]
one has $h_1(S)=0$ and $h_3(S)=(1+i)(1-i)=-1$, so $h_3+h_1+1=0$.
\end{proof}

\subsubsection{The fields \texorpdfstring{$\mathbb{F}_{25}$ and $\mathbb{F}_{27}$}{F25 and F27}}

For $q=p^e$, write
\[
\operatorname{Tr}(x)
:=\operatorname{Tr}_{\mathbb{F}_q/\mathbb{F}_p}(x)
=x+x^p+\cdots+x^{p^{e-1}}.
\]

\begin{proposition}\label{prop:q25q27}
For $q=25$, every target in $\mathbb{F}_{25}$ is the sum of an $s$-subset of $D$ for $s=3,4,5$.  For $q=27$, the same holds for $s=3,4,5,6$.
\end{proposition}

\begin{proof}
We use \eqref{eq:FourierNs} and the coefficient recurrence displayed below.

For $q=25$, $\lvert D\rvert=11$.  Let $\zeta=\zeta_5$, let $\chi$ be the quadratic character, and for $t\ne0$ put
\[
\psi_t(x)=\zeta^{\operatorname{Tr}(tx)},
\qquad \varepsilon=\chi(t),
\qquad u=\operatorname{Tr}(t)\in\mathbb{F}_5.
\]
By \cite[Thm.~5.15]{LidlNiederreiter1997}, the canonical quadratic Gauss sum over $\mathbb{F}_{25}$ is $-5$.  Since each $j\in\mathbb{F}_5^\times$ is a square in $\mathbb{F}_{25}$,
\[
A_j(t)=\sum_{x\in D}\psi_t(jx)
=-\frac{1+5\varepsilon}{2}-\zeta^{ju},\qquad 1\leqslant j\leqslant4,
\]
and $A_5(t)=11$.  The recurrence
\[
s e_s=\sum_{j=1}^s(-1)^{j-1}e_{s-j}A_j
\]
gives, for $u\ne0$ and $\tau_u=\zeta^u+\zeta^{-u}$,
\[
\begin{array}{c|ccc}
&\lvert e_3\rvert^2&\lvert e_4\rvert^2&\lvert e_5\rvert^2\\ \hline
\varepsilon=-1&5(1-\tau_u)&5(1-\tau_u)&7+\tau_u\\
\varepsilon=+1&100+35\tau_u&190+65\tau_u&262+101\tau_u.
\end{array}
\]
For $u=0$, all four nonzero trace-zero elements are nonsquares and $e_3=e_4=0$, $\lvert e_5\rvert=2$.  Each nonzero trace fiber contains three squares and two nonsquares.  Since
\[
\tau_{\pm1}=\frac{\sqrt5-1}{2}\in(0,1),
\qquad
\tau_{\pm2}=-\frac{\sqrt5+1}{2}\in(-2,0),
\]
we obtain
\[
\sum_{t\ne0}\lvert e_3(t)\rvert<160<\binom{11}{3},
\quad
\sum_{t\ne0}\lvert e_4(t)\rvert<208<\binom{11}{4},
\]
\[
\sum_{t\ne0}\lvert e_5(t)\rvert<254<\binom{11}{5}.
\]
Equation~\eqref{eq:FourierNs} therefore gives the assertion for $q=25$.

For $q=27$, $\lvert D\rvert=12$.  Let $\zeta=e^{2\pi i/3}$, $u=\operatorname{Tr}(t)\in\mathbb{F}_3$, and $\varepsilon=\chi(t)$.  By \cite[Thm.~5.15]{LidlNiederreiter1997}, the canonical quadratic Gauss sum over $\mathbb{F}_{27}$ is $G=-3i\sqrt3$.  Since the quadratic character of $\mathbb{F}_{27}$ restricts to the quadratic character of $\mathbb{F}_3$, for $3\nmid j$ we have
\[
A_j(t)=\frac{-1+\chi(jt)G}{2}-\zeta^{ju},
\]
whereas $A_j(t)=12$ for $3\mid j$.  The same recurrence then yields
\[
\begin{array}{c|c}
(u,\varepsilon)&(\lvert e_3\rvert^2,\lvert e_4\rvert^2,\lvert e_5\rvert^2,\lvert e_6\rvert^2)\\ \hline
(0,\pm1)&(16,81,81,36)\\
(1,+1),(2,-1)&(52,27,108,225)\\
(1,-1),(2,+1)&(7,27,27,9).
\end{array}
\]
Thus $\lvert e_3\rvert<8$, $\lvert e_4\rvert \leqslant9$, $\lvert e_5\rvert<11$, and $\lvert e_6\rvert \leqslant15$.  There are $26$ nontrivial additive characters, and
\[
\binom{12}{3}>26\cdot8,
\quad \binom{12}{4}>26\cdot9,
\quad \binom{12}{5}>26\cdot11,
\quad \binom{12}{6}>26\cdot15.
\]
Again \eqref{eq:FourierNs} gives every target.
\end{proof}

\begin{theorem}[The case $m=3$ in odd characteristic]\label{thm:odd}
The case $m=3$ in Theorem~\ref{thm:main}~(\ref{it:small-excess-classification}) holds for odd $q$.
\end{theorem}
\begin{proof}
Apply the external ranges \eqref{eq:CLZ-range} and
\eqref{eq:Kaipa-range}.  For $q\geqslant121$, the low range follows
from Theorem~\ref{thm:low-dimension-gcd} with $d=4$:
Fan's condition holds and $\gcd(4,q-1)>1$.
Failure of $q>4(k-1)+1$ gives $k+1>q/8$, while below
Kaipa's range we have $k+1\leqslant q/2$.  The remaining range
is therefore covered by Lemma~\ref{lem:odd-middle-count}.

For $q=49,81$, the low/high alternatives leave respectively
$n=k+1=14,\ldots,24$ and $n=22,\ldots,40$.
Here $M=2\sqrt q$ and $(q-M)/p$ are integers, so
$A_{3,n}$ is an exact integer coefficient.  The minima of
$\binom qn-q^2 A_{3,n}$ over these respective intervals are
$626158230621$ and $37184292070905203010$, both positive;
the ancillary program evaluates these convolutions using
exact rational arithmetic.

The remaining nonprime odd fields are $q=9,25,27$.
Proposition~\ref{prop:q9} handles $(9,3)$.
For $q=25$, the residual $k=7,\ldots,11$ has
$s=\lfloor k/2\rfloor\in\{3,4,5\}$, and for $q=27$,
$k=8,\ldots,12$ has $s\in\{4,5,6\}$.
Proposition~\ref{prop:q25q27}, together with
Lemma~\ref{lem:one-core} and the zero-coefficient branch
Lemma~\ref{lem:easy-branch}, supplies the required agreement sets.
Every prime field is already in \eqref{eq:CLZ-range}.
These alternatives cover every odd prime power.
\end{proof}

\subsection{Characteristic two: Artin--Schreier fibers}\label{sec:char2}

Assume now that $q=2^e$.  The Kaipa range \eqref{eq:Kaipa-range} becomes
\[
k\geqslant q/2-1.
\]
The additional even-characteristic nucleus possibility in \cite{Kaipa2017} occurs when $n-k=3$; at full length $n=q$ this is $k=q-3$, outside the admissible degree range $k\leqslant q-4$.

Let
\[
V:=\ker\operatorname{Tr}_{\mathbb{F}_q/\mathbb{F}_2},\qquad N:=\lvert V\rvert=q/2.
\]
By \cite[Thm.~2.25]{LidlNiederreiter1997}, the \emph{Artin--Schreier map} $w\mapsto w^2+w$ has image $V$ and every fiber is $\{w,w+1\}$.  For $\delta\ne0$ and $t\in V$, define
\[
P_\delta(t):=\{\delta w,\delta(w+1)\},\qquad w^2+w=t.
\]
The fiber description above shows that $P_\delta(t)$ is independent of the choice of the root $w$.  Distinct $t$ give disjoint pairs, and for fixed $\delta$ these pairs partition $\mathbb{F}_q$.

\begin{lemma}[Common-difference block identities]\label{lem:block-identities}
Take $R$ distinct parameters $t_1,\dots,t_R\in V$ and set $T:=\sum_i t_i$.  For the union of the blocks $P_\delta(t_i)$,
\begin{equation}\label{eq:block-h}
\begin{aligned}
h_1&=\delta\varepsilon_1,\\
h_2&=\delta^2(T+\varepsilon_2),\\
h_3&=\delta^3(\varepsilon_TT+\varepsilon_3),
\end{aligned}
\end{equation}
where
\[
\varepsilon_1\equiv R,\quad
\varepsilon_2\equiv\binom{R+1}{2},\quad
\varepsilon_T\equiv R+1,\quad
\varepsilon_3\equiv\binom{R+2}{3}\pmod2.
\]
\end{lemma}

\begin{proof}
For one normalized block $\{w,w+1\}$ with $w^2+w=t$, a direct calculation gives $h_1=1$, $h_2=1+t$, and $h_3=1$.  Multiplying the generating functions
\[
\prod_{x\in P_\delta(t_i)}(1-xz)^{-1}
\]
and extracting degrees at most three gives \eqref{eq:block-h}.
\end{proof}

\subsubsection{Exact subset sums in the trace hyperplane}

We record the elementary character-sum identities that will be used below.

\begin{lemma}[Trace-hyperplane subset sums]\label{lem:V-subsets}
Let $V$ be an $\mathbb{F}_2$-vector space of size $N=2^d$.
\begin{enumerate}[label=(\arabic*),ref=\arabic*]
\item\label{item:V-all} If $N\geqslant16$, every element of $V$ is the sum of an $r$-subset of $V$ for $3\leqslant r\leqslant N-3$.
\item\label{item:V-nonzero} If $N\geqslant16$, every element of $V$ is the sum of an $r$-subset of $V\setminus\{0\}$ for $3\leqslant r\leqslant N-4$.
\item\label{item:V-two-delete} If $N\geqslant32$ and $\tau\ne0$, every element of $V$ is the sum of an $r$-subset of $V\setminus\{0,\tau\}$ for $3\leqslant r\leqslant N-5$.
\end{enumerate}
The exact Fourier formulas below also give the exceptional counts needed when $N=8$ and $N=16$.
\end{lemma}

\begin{proof}
Put $M=N/2$.  For a nontrivial additive character $\chi$ of $V$, half the elements have value $+1$ and half $-1$, so
\[
\prod_{x\in V}(1+z\chi(x))=(1-z^2)^M.
\]
Let $M_V(r,\beta)$ denote the number of $r$-subsets of $V$ with sum $\beta$.  Character orthogonality gives
\[
N M_V(r,\beta)=\binom Nr+\sigma_\beta C_r,
\]
where $\sigma_0=N-1$, $\sigma_\beta=-1$ for $\beta\ne0$, and
\[
C_{2j}=(-1)^j\binom Mj,\qquad C_{2j+1}=0.
\]
For odd $r$ the correction vanishes.  For $r=2j\leqslant N/2$ with $r\geqslant4$, the binomial convolution gives
\[
\binom{2M}{2j}\geqslant\binom Mj^2>(2M-1)\binom Mj,
\]
because $M\geqslant8$ and $\binom Mj\geqslant\binom M2>2M-1$.  Thus every target occurs for $3\leqslant r\leqslant N/2$.  Since the sum of all elements of $V$ is $0$, complementation gives the remaining sizes up to $N-3$.  This proves~(\ref{item:V-all}).

For $D_0=V\setminus\{0\}$, a nontrivial character has generating product
\[
(1-z^2)^{M-1}(1-z).
\]
If $C'_r$ is its coefficient of $z^r$, then
\[
\lvert C'_r\rvert=\binom{M-1}{\lfloor r/2\rfloor}
\]
and
\[
N M_{D_0}(r,\beta)=\binom{N-1}{r}+
\begin{cases}
(N-1)C'_r,&\beta=0,\\
-C'_r,&\beta\ne0.
\end{cases}
\]
For $r\leqslant M-1$, write $j=\lfloor r/2\rfloor$.  If $r=2j$, the central term in the binomial convolution gives
\[
\binom{2M-1}{2j}\geqslant\binom Mj\binom{M-1}{j},
\]
while if $r=2j+1$ it gives
\[
\binom{2M-1}{2j+1}\geqslant\binom M{j+1}\binom{M-1}{j}.
\]
After division by $\lvert C'_r\rvert$, either lower bound is at least $\binom M2>2M-1$ for $r\geqslant3$.  Hence all targets occur up to the midpoint.  The sum of $D_0$ is $0$, so complementation proves~(\ref{item:V-nonzero}).

Finally fix $\tau\ne0$ and put $D_\tau=V\setminus\{0,\tau\}$.  Split the nontrivial characters according as $\chi(\tau)=+1$ or $-1$.  Their generating products are
\[
P_+(z)=(1-z^2)^{M-2}(1-z)^2,
\qquad
P_-(z)=(1-z^2)^{M-1}.
\]
If $E_r=[z^r]P_+(z)$ and $F_r=[z^r]P_-(z)$, then
\[
E_{2j}=(-1)^j\left[\binom{M-2}{j}-\binom{M-2}{j-1}\right],
\quad
E_{2j+1}=-2(-1)^j\binom{M-2}{j},
\]
\[
F_{2j}=(-1)^j\binom{M-1}{j},\qquad F_{2j+1}=0.
\]
The sums of the characters in the two classes are
\[
(S_+,S_-)=(M-1,M),\ (M-1,-M),\ (-1,0)
\]
according as $\beta=0$, $\beta=\tau$, or $\beta\notin\{0,\tau\}$.  Hence
\[
N M_{D_\tau}(r,\beta)=\binom{N-2}{r}+S_+(\beta)E_r+S_-(\beta)F_r.
\]
For $3\leqslant r\leqslant(N-2)/2$, with $j=\lfloor r/2\rfloor$, the correction has absolute value less than
\[
3M\binom{M-1}{j}.
\]
For even $r=2j$, the binomial convolution contains $\binom{M-1}{j}^2$; for odd $r=2j+1$, it contributes $\binom{M-1}{j}\binom{M-1}{j+1}$.  Therefore
\[
\frac{\binom{N-2}{r}}{\binom{M-1}{j}}\geqslant\binom{M-1}{2}>3M
\]
when $M\geqslant16$.  Complementation sends an $r$-subset sum $\beta$ to an $(N-2-r)$-subset sum $\tau+\beta$, because the sum of $D_\tau$ is $\tau$.  This proves~(\ref{item:V-two-delete}).

For the two smaller vector spaces actually used below, the same formulas give the needed exceptional counts explicitly.  If $N=8$, the number of three-subsets of $V$ with sum $\beta$ is $7$ for every $\beta$; the number of three-subsets of $V\setminus\{0\}$ with sum $\beta$ is $7$ when $\beta=0$ and $4$ otherwise; and the number of four-subsets of $V$ with sum $\beta$ is $14$ when $\beta=0$ and $8$ otherwise.  If $N=16$, the number of $r$-subsets of $D_\tau$ with sum $\beta$ is $77,56,62$ for $r=4$, according as $\beta=0,\tau$, or otherwise, and $168,203,188$ for $r=6$ in the same three cases.  All these counts are positive.
\end{proof}

\subsubsection{Trace selection}

The block identities prescribe a value of $T$, which must lie in $V$.  The next three lemmas ensure that this value has trace zero.

\begin{lemma}[First trace-zero selection lemma]\label{lem:T1}
Let $q=2^e\geqslant16$, $a,c\in\mathbb{F}_q$, and $\varepsilon\in\mathbb{F}_2$.  There exists
\[
\delta\in\mathbb{F}_q^\times\setminus\{a\}
\]
such that
\[
\operatorname{Tr}\left(\frac{a\varepsilon+c\delta^{-2}}{a+\delta}\right)=0.
\]
If $c\ne0$ or $a\varepsilon\ne0$, $\delta$ may be chosen so that the displayed target is nonzero.
\end{lemma}

\begin{proof}
If $a=0$, the target is $c\delta^{-3}$.  For $c\ne0$, put $x=\delta^{-1}$.  After choosing an $\mathbb{F}_2$-basis of $\mathbb{F}_q$, the map $x\mapsto\operatorname{Tr}(cx^3)$ is a polynomial of total degree at most two in $e$ variables and vanishes at $x=0$.  Since $e\geqslant4$, Chevalley's theorem \cite[Cor.~6.6]{LidlNiederreiter1997} gives a nonzero zero.  The case $c=0$ is immediate, and if $c\ne0$ the target is automatically nonzero.

Assume $a\ne0$ and write $\delta=az$, $C=c/a^3$.  Then
\[
\frac{\varepsilon+Cz^{-2}}{1+z}
\]
has trace
\[
\operatorname{Tr}\left(\frac A z+\frac B{1+z}\right),
\qquad A=C+C^{1/2},\quad B=C+\varepsilon.
\]
The substitution $u=z/(1+z)$ is a bijection of $\mathbb{F}_q\setminus\{0,1\}$ with itself, and the trace condition becomes
\[
\operatorname{Tr}(A/u+Bu+A+B)=0.
\]
If $A$ and $B$ are nonzero, the Kloosterman-sum estimate \cite[Def.~5.42 and Thm.~5.45]{LidlNiederreiter1997} shows that each prescribed trace class contains at least
\[
\frac{q-1-2\sqrt q}{2}
\]
values of $u\in\mathbb{F}_q^\times$ before excluding $u=1$.  If exactly one of $A,B$ is nonzero, a nonzero linear trace equation in $u$ or $u^{-1}$ has at least $q/2-1$ nonzero solutions; if both vanish, the trace condition is automatic.  These counts leave an admissible $u\ne1$ for $q\geqslant16$.

For the refinement, the target can vanish only when
\[
a\varepsilon+c\delta^{-2}=0.
\]
Under the hypothesis $c\ne0$ or $a\varepsilon\ne0$, this equation is not identically zero and has at most one nonzero solution, because squaring is bijective in characteristic two.  Removing that value from the preceding trace-compatible choices still leaves an admissible $\delta$.
\end{proof}

\begin{lemma}[Second trace-zero selection lemma]\label{lem:T2}
Suppose $q=2^e\geqslant16$, $R$ is odd, and $a\ne0$.  With the parities from Lemma~\ref{lem:block-identities}, the target
\[
\beta(\delta):=\frac{\varepsilon_3}{a}\delta+\varepsilon_2+\frac ba\delta^{-1}+\frac ca\delta^{-2}
\]
can be made trace zero by a suitable $\delta\ne0$.
\end{lemma}

\begin{proof}
For odd $R$, $\varepsilon_2=\varepsilon_3$.  Put
\[
\Delta_0=\frac ba+\left(\frac ca\right)^{1/2}.
\]
Then
\[
\operatorname{Tr}(\beta(\delta))=
\operatorname{Tr}\left(\frac{\varepsilon_3}{a}\delta+\Delta_0\delta^{-1}+\varepsilon_2\right).
\]
If the two nonconstant coefficients are nonzero, the Kloosterman-sum estimate \cite[Thm.~5.45]{LidlNiederreiter1997} supplies a suitable $\delta$.  If exactly one is nonzero, use a nonzero linear trace equation in $\delta$ or $\delta^{-1}$.  If both vanish, then $\varepsilon_3=\varepsilon_2=0$ and the trace is already zero.
\end{proof}

\begin{lemma}[Exceptional trace-zero selection lemma]\label{lem:T3}
Suppose $q=2^e\geqslant16$, $R$ is odd, and $a=0$.  Put $r=R-1$.  There is a choice $\eta\notin\mathbb{F}_2$ and $\delta\ne0$ such that, for
\[
\tau=\eta^2+\eta,
\qquad D_\tau=V\setminus\{0,\tau\},
\]
the construction using $r$ common-difference blocks together with the singleton points $0$ and $\eta\delta$ prescribes a target $\beta(\delta)\in V$ for the sum of the $r$ block parameters.
\end{lemma}

\begin{proof}
The coefficient of the subset-sum parameter is $\delta^3(1+\eta)\ne0$, and the target is
\[
\beta(\delta)=
\frac{\varepsilon_3+\eta\varepsilon_2+\eta^3}{1+\eta}
+\frac{b\eta}{1+\eta}\delta^{-2}
+\frac{c}{1+\eta}\delta^{-3}.
\]
Since $r$ is even, $\varepsilon_3=0$, and $\varepsilon_2=1$ exactly when $r\equiv2\pmod4$.  In that case the constant term is $\eta^2+\eta\in V$.  If $r\equiv0\pmod4$, write $u=(1+\eta)^{-1}$; then
\[
\operatorname{Tr}\frac{\eta^3}{1+\eta}=\operatorname{Tr}(1+u).
\]
Choose $u\in\mathbb{F}_q^\times\setminus\{1\}$ with $\operatorname{Tr}(u)=\operatorname{Tr}(1)$, which is possible for $q\geqslant16$.  Thus the constant term has trace zero.

Now put $x=\delta^{-1}$.  After choosing an $\mathbb{F}_2$-basis of $\mathbb{F}_q$, the trace of $\beta$ is a polynomial of total degree at most two in $e$ variables and vanishes at the origin.  Since $e\geqslant4$, Chevalley's theorem \cite[Cor.~6.6]{LidlNiederreiter1997} gives a nonzero zero.  Taking its inverse gives the desired $\delta$.
\end{proof}

\subsubsection{The middle dimensions}

\begin{proposition}[Middle range]\label{prop:even-middle}
Let $q=2^e\geqslant16$ and $6\leqslant k<q/2$.  Then there is a $(k+1)$-subset $S\subseteq\mathbb{F}_q$ with $H_f(S)=0$.
\end{proposition}

\begin{proof}
Write the required cardinality in Artin--Schreier blocks.

Suppose first that $k=2R$ is even.  Use $R$ blocks and one singleton $\eta\delta$, where $\eta\equiv R\pmod2$ and $\eta\in\{0,1\}$.  Thus the singleton is $0$ when $R$ is even and $\delta$ when $R$ is odd.  Lemma~\ref{lem:block-identities}, together with $\varepsilon_2=\varepsilon_3$ for odd $R$, gives for the resulting $(2R+1)$-set
\[
h_1=0,
\qquad h_2=\delta^2(T+\varepsilon_2),
\qquad h_3=\delta^3T.
\]
Hence
\[
\delta^2(a+\delta)T+a\varepsilon_2\delta^2+c=0.
\]
Lemma~\ref{lem:T1} with $\varepsilon=\varepsilon_2$ chooses $\delta\ne0,a$ so that the prescribed value of $T$ lies in $V$.  A block meets the singleton precisely when its parameter is $0$, so the $R$ distinct parameters are chosen in $V\setminus\{0\}$; Lemma~\ref{lem:V-subsets} supplies them.

Now let $k$ be odd and put $R=(k+1)/2$, so the construction uses $R$ blocks and no singleton.  If $R$ is even, then $\varepsilon_1=\varepsilon_3=0$ and $\varepsilon_T=1$, so the same displayed formulas hold; Lemma~\ref{lem:T1} again makes the prescribed $T\in V$, and Lemma~\ref{lem:V-subsets}~(\ref{item:V-all}) supplies the $R$ parameters.

If $R$ is odd and $a\ne0$, Lemma~\ref{lem:block-identities} gives
\[
H_f=a\delta^2\bigl(T+\beta(\delta)\bigr),
\]
with $\beta(\delta)$ exactly as in Lemma~\ref{lem:T2}.  Choose $\delta$ there and then use an $R$-subset of all of $V$ with sum $\beta(\delta)$.

Finally suppose that $R$ is odd and $a=0$.  Put $r=R-1$.  Lemma~\ref{lem:T3} uses $r$ common-difference blocks together with the singleton points $0$ and $\eta\delta$.  A block meets one of these two singleton points exactly for parameter $0$ or $\tau=\eta^2+\eta$, respectively.  Thus the $r$ parameters lie in $D_\tau=V\setminus\{0,\tau\}$, and Lemma~\ref{lem:V-subsets}~(\ref{item:V-two-delete}) supplies them once the trace-zero target from Lemma~\ref{lem:T3} is fixed.

For $q\geqslant64$, the stated ranges in Lemma~\ref{lem:V-subsets} contain all the subset sizes above: in the even case $3\leqslant R\leqslant N/2-1$, and in the exceptional two-deletion case $4\leqslant r\leqslant N/2-1$.  When $q=16$, the only additional sizes in this proposition are the $N=8$ cases $\lvert V\setminus\{0\}\rvert$ with $r=3$ and $V$ with $r=4$, both covered by the explicit counts at the end of Lemma~\ref{lem:V-subsets}.  When $q=32$, Lemma~\ref{lem:V-subsets}~(\ref{item:V-all}) and~(\ref{item:V-nonzero}) already apply because $N=16$; only the exceptional two-deletion sizes $r=4,6$ fall outside~(\ref{item:V-two-delete}), and these are exactly the two positive $N=16$ counts recorded there.
\end{proof}

\subsubsection{The dimensions \texorpdfstring{$k\leqslant5$}{k <= 5}}

\begin{lemma}[The case $k=2$]\label{lem:k2-even}
For every even $q\geqslant8$, there is a three-set $S$ with $H_f(S)=0$.
\end{lemma}

\begin{proof}
For a three-set $X$, translation by $t$ gives
\[
h_1(X+t)=h_1(X)+t,
\quad h_2(X+t)=h_2(X),
\quad h_3(X+t)=h_3(X)+t h_2(X),
\]
so
\[
H_f(X+t)=H_f(X)+t(h_2(X)+b).
\]
Take $X=\{0,\lambda,\lambda r\}$ with $r\notin\{0,1\}$ and $1+r+r^2\ne0$.  Such an $r$ exists because at most four field elements are excluded and $q\geqslant8$.  Then
\[
h_2(X)=\lambda^2(1+r+r^2).
\]
As $\lambda$ ranges over $\mathbb{F}_q^\times$, this expression ranges over all nonzero field elements because squaring is bijective.  Choose $\lambda$ so that $h_2(X)\ne b$, and solve uniquely for $t$.
\end{proof}

\begin{lemma}[Two Artin--Schreier pairs]\label{lem:transverse}
Let $q\geqslant16$ be even.  For two distinct nonzero differences $\delta_1,\delta_2$, the union of two Artin--Schreier pairs has
\[
H_f=H_0+\gamma_1t_1+\gamma_2t_2,
\]
where
\[
\gamma_1:=\delta_1^2(\delta_2+a),
\qquad
\gamma_2:=\delta_2^2(\delta_1+a).
\]
The differences can be chosen so that $\gamma_1\gamma_2\ne0$ and $\gamma_1\ne\gamma_2$.  Then every target has exactly $q/4$ solutions $(t_1,t_2)\in V^2$, and at most two of them yield intersecting pairs.
\end{lemma}

\begin{proof}
Choose $\delta_1\notin\{0,a\}$ and choose $\delta_2$ outside
\[
\left\{0,\delta_1,a,\frac{a\delta_1}{a+\delta_1}\right\}.
\]
Then $\gamma_1\gamma_2\ne0$ and $\gamma_1\ne\gamma_2$.  If a nonzero scalar $\lambda$ satisfies $\lambda V=V$, then the two nonzero $\mathbb{F}_2$-linear functionals $x\mapsto\operatorname{Tr}(x)$ and $x\mapsto\operatorname{Tr}(\lambda^{-1}x)$ have the same kernel and hence are equal.  By \cite[Thm.~2.24]{LidlNiederreiter1997}, this forces $\lambda=1$.  Thus $\gamma_1V$ and $\gamma_2V$ are distinct hyperplanes and their sum is $\mathbb{F}_q$, giving exactly $q/4$ solutions.  If the two pairs meet at $x$, substitution reduces to a quadratic equation
\[
(\delta_1+\delta_2)(x^2+ax)=D'
\]
for a fixed $D'$, hence to at most two intersecting solutions.
\end{proof}

\begin{lemma}[The cases $k=3,4$]\label{lem:k34-even}
For $q\geqslant16$, the case $k=3$ is solved by two pairs as in Lemma~\ref{lem:transverse}.  For $q\geqslant32$, the case $k=4$ is solved by two such pairs together with the point $0$.  The remaining case $q=16,k=4$ is also solvable.
\end{lemma}

\begin{proof}
For $k=3$, Lemma~\ref{lem:transverse} gives $q/4$ affine solutions and at most two intersecting ones, so a disjoint solution exists.

For $k=4$, add $0$.  Besides pair intersections, exclude $t_1=0$ and $t_2=0$, because $P_{\delta_i}(0)$ contains $0$.  Thus there are at most four bad solutions; $q/4>4$ for $q\geqslant32$.

For $q=16,k=4$, use two common-difference blocks plus $0$.  They satisfy
\[
h_1=0,
\qquad h_2=\delta^2(T+1),
\qquad h_3=\delta^3T.
\]
If $(a,c)\ne(0,0)$, Lemma~\ref{lem:T1} with $\varepsilon=1$ gives $T\in V\setminus\{0\}$; since $\lvert V\rvert=8$, write $T=t_1+t_2$ with distinct nonzero $t_i$.  If $(a,c)=(0,0)$, scaling reduces $b$ to $0$ or $1$.  Let $\theta^4=\theta+1$ in $\mathbb{F}_{16}$.  The sets
\[
S_0=\{0,1,\theta,\theta^2,\theta^2+\theta^3\}
\]
and
\[
S_1=\{0,1,\theta,1+\theta,\theta+\theta^2\}
\]
satisfy, respectively, $h_3(S_0)=0$ and $h_1(S_1)=h_3(S_1)=\theta+\theta^2$.  They solve $b=0$ and $b=1$.
\end{proof}

\begin{lemma}[The case $k=5$]\label{lem:k5-even}
For every even $q\geqslant16$, there is a six-set $S$ with $H_f(S)=0$.
\end{lemma}

\begin{proof}
If $a\ne0$, use three common-difference blocks.  With $T=t_1+t_2+t_3$,
\[
h_1=\delta,
\qquad h_2=\delta^2T,
\qquad h_3=0,
\]
so the target is
\[
T=\frac ba\delta^{-1}+\frac ca\delta^{-2}.
\]
Put
\[
\Delta_0=\frac ba+\left(\frac ca\right)^{1/2}.
\]
Then $\operatorname{Tr}(T)=\operatorname{Tr}(\Delta_0\delta^{-1})$.  If $\Delta_0=0$, every nonzero $\delta$ gives $\operatorname{Tr}(T)=0$.  If $\Delta_0\ne0$, choose a nonzero $x$ in the kernel of the nonzero $\mathbb{F}_2$-linear functional $x\mapsto\operatorname{Tr}(\Delta_0x)$ and set $\delta=x^{-1}$.  Thus $T\in V$.  If $\lvert V\rvert \geqslant16$, Lemma~\ref{lem:V-subsets}~(\ref{item:V-all}) supplies three distinct summands; when $\lvert V\rvert=8$, the explicit three-subset count in its proof does the same.

Let $a=0$.  With three common-difference blocks, $H_f=b\delta+c$.  If $bc\ne0$, take $\delta=c/b$; if $b=c=0$, any $\delta$ works.

It remains to treat the two cases in which exactly one of $b,c$ is nonzero.  Use two blocks of difference $\delta$ and a third of difference $\lambda\delta$, where $\lambda\notin\{0,1\}$.  Put
\[
U=t_1+t_2,\qquad \tau=\lambda^2+\lambda.
\]
Multiplying the three two-point generating functions through degree three gives
\[
h_1=\lambda\delta,
\qquad
h_3=\delta^3(1+\lambda)(U+\tau),
\]
and the third block parameter $t_3$ does not occur in either expression.

Suppose first that $b\ne0$ and $c=0$.  The equation $H_f=0$ becomes
\[
U=\tau+\frac{b\lambda}{1+\lambda}\delta^{-2}.
\]
Choose $y\in V^\times\setminus\{\tau\}$ and then choose $\delta$ so that the second term equals $y$; squaring is bijective.  Then $U=\tau+y\ne0$, so it is the sum of two distinct elements $t_1,t_2\in V$.

Finally suppose $b=0$ and $c\ne0$.  We seek
\[
U=\tau+v,
\qquad
v=\frac{c}{(1+\lambda)\delta^3}\in V^\times,
\qquad v\ne\tau.
\]
Fix any $v\in V^\times$ and put $s_0=1+\lambda$.  The condition that one can solve
\[
\delta^3=\frac{c}{s_0v}
\]
is exactly that $s_0$ lie in the coset $(c/v)(\mathbb{F}_q^\times)^3$, whose size is
\[
\frac{q-1}{\gcd(3,q-1)}\geqslant5.
\]
Exclude $s_0=1$ and the at most two roots of $s_0^2+s_0=v$; a choice remains.  It gives $\lambda\notin\{0,1\}$, $v\ne\tau$, and a suitable $\delta$.  Hence again $U=\tau+v\ne0$ is the sum of two distinct parameters $t_1,t_2\in V$.

In either case, after the first two disjoint pairs are chosen, at most four parameters of the partition into pairs of difference $\lambda\delta$ yield a pair meeting one of their four points.  Since $\lvert V\rvert \geqslant8$, choose $t_3$ outside those parameters.  The resulting three blocks are pairwise disjoint, and the choice of $t_3$ does not change $H_f$.
\end{proof}

\subsection{Completion of the case \texorpdfstring{$m=3$}{m=3}}\label{sec:global}

\begin{theorem}\label{thm:even}
The case $m=3$ in Theorem~\ref{thm:main}~(\ref{it:small-excess-classification}) holds when $q$ is even.
\end{theorem}

\begin{proof}
Apply Kaipa first in the range $k\geqslant q/2-1$.  Its additional even-characteristic case with $n-k=3$ occurs only at $k=q-3$, outside the present range $k\leqslant q-4$.

If $q=8$, only $k=2$ remains below the Kaipa range; Lemma~\ref{lem:k2-even} treats this case.  If $q=16$, use Lemmas~\ref{lem:k2-even},~\ref{lem:k34-even}, and~\ref{lem:k5-even} for $k=2,3,4,5$, Proposition~\ref{prop:even-middle} for $k=6$, and Kaipa for $k\geqslant7$.  If $q\geqslant32$, use the same small-$k$ lemmas for $k=2,3,4,5$, Proposition~\ref{prop:even-middle} for $6\leqslant k\leqslant q/2-2$, and Kaipa for $k\geqslant q/2-1$.  Each of these constructions gives a $(k+1)$-subset with $H_f(S)=0$, so Proposition~\ref{prop:agreement} yields \eqref{eq:mainbound}; the Kaipa range gives the same conclusion directly because the distance is integral and strictly below the covering radius.
\end{proof}

\begin{proof}[Completion of the case $m=3$ in Theorem~\ref{thm:main}~(\ref{it:small-excess-classification})]
Theorems~\ref{thm:odd} and~\ref{thm:even} cover the odd- and even-characteristic cases.
\end{proof}

\section{Exact classification when \texorpdfstring{$m=4$}{m=4}}\label{sec:excess-four}
\subsection{Odd characteristic}\label{sec:excess-four-odd}
Let
\[
 F=X^{k+4}+aX^{k+3}+bX^{k+2}+cX^{k+1}+dX^k
\]
be the normalized interpolation representative.  By
Lemma~\ref{lem:fourth-remainder}, its agreement condition is
$h_4+a h_3+b h_2+c h_1+d=0$ on a $(k+1)$-subset.
The generic low-dimensional bound is included for small
characteristics for which a degree-five exceptional-polynomial
classification alone would not suffice.
The low-dimensional argument uses the following quotient formulation
of the same agreement condition.  For an arbitrary set $A$ with
$\lvert A\rvert=k-1$ write
\begin{equation}\label{eq:quintic-quotient}
 P_A(X)=\prod_{t\in A}(X-t),\qquad
 Q_A(X)=\frac{F(X)-(F\bmod P_A)(X)}{P_A(X)}.
\end{equation}
It is monic of degree five.  If the received word is a deep hole,
$Q_A$ is injective on $\mathbb F_q\setminus A$: a collision would
produce a degree-$<k$ word agreeing at $A$ and two further points.
By the polynomial value-set bound \cite{MullenWanWang2013},
nonpermutation quintics have image of size at most
$q-\lceil(q-1)/5\rceil$.  Therefore the strict inequality
$q>5k-4$ forces \emph{every} $Q_A$ to permute $\mathbb F_q$.

\begin{proposition}[Low-dimensional exclusion]\label{prop:odd-low}
Let $q$ be an odd prime power, $q\geqslant11$, and
$2\leqslant k\leqslant q-5$.  If $q>5k-4$, then no polynomial of
interpolation degree $k+4$ is a deep hole.
\end{proposition}
\begin{proof}
Fix $V\subseteq\mathbb F_q$ with $\lvert V\rvert=k-2$ and put
$R=(F-(F\bmod P_V))/P_V$, a monic sextic.  For each
$u\notin V$, the polynomial
$Q_{V\cup\{u\}}=(R(X)-R(u))/(X-u)$ must be a permutation.
For $1\leqslant n\leqslant q-2$ satisfying $q-1\leqslant5n<2(q-1)$, Hermite's
power-sum criterion forces
\[
 C_n(u)=[X^{q-1}]Q_{V\cup\{u\}}(X)^n=0
 \qquad(u\notin V).
\]
Put $r=5n-(q-1)$.  Then $\deg C_n\leqslant r$, and the coefficient
of $u^r$ is
\begin{equation}\label{eq:hermite-obstruction}
 [t^r](1+t+\cdots+t^5)^n.
\end{equation}
The coefficient table below gives, in each characteristic, a choice
of $n$ for which \eqref{eq:hermite-obstruction} is nonzero and
$r<q-k+2$.  Thus $C_n$ cannot vanish at all $q-k+2>r$
possible values of $u$.  This contradiction proves the claim.
The table entries are exact residues of binomial coefficients;
they are also reproduced by the verification program.
\end{proof}

\begin{center}
\small
\begin{tabular}{@{}lll@{}}
\hline
Characteristic / residue class & choice & nonzero coefficient \\
\hline
$p\notin\{2,3,5,7,13\}$ & $n=\lceil(q-1)/5\rceil$ & $1,\ 3/25,\ 28/125,\ 234/625$ \\
$p=3,\ q\bmod5=3,2,4$ & $n=j,j+1,j+2$ & $2,1,1$ \\
$p=7,\ q\bmod5=3$ & $n=j+3$ & $2$ \\
$p=13,\ q\bmod5=2$ & $n=j+5$ & $9$ \\
$p=5$ & $n=(q-1)/4$ & $1$ \\
\hline
\end{tabular}
\end{center}
\noindent
Here $j=\lceil(q-1)/5\rceil$; the small-characteristic rows
use the corresponding $p$-power residue calculations, including
$n\bmod 9,27,49,169$, and the last row follows by base-five
Frobenius expansion.  The ordinary rows include $q\bmod5=1,4,3,2$
respectively; the displayed coefficient list must be read in
that order.  For $p=3,7,13$, residue classes not listed separately
also use $n=j$ and the corresponding ordinary coefficient.
The ancillary program \nolinkurl{core/verify_small_cases.py} reproduces
the rational coefficients and the finite modular table.

\begin{theorem}[The case $m=4$ in odd characteristic]
\label{thm:four-odd}
If $q$ is an odd prime power and $2\leqslant k\leqslant q-5$, then
no received word of interpolation degree $k+4$ is a deep hole of
$\operatorname{RS}_q(k)$.
\end{theorem}
\begin{proof}
Use \eqref{eq:CLZ-range} and \eqref{eq:Kaipa-range} first.
If $q\geqslant211$, Proposition~\ref{prop:odd-low} excludes
$q>5k-4$.  Otherwise $k+1>q/8$, and the high-dimensional
complement gives $k+1\leqslant q/2$; hence
Lemma~\ref{lem:odd-middle-count} applies with $m=4$.

For the proper odd prime powers below $211$, the same
low/high alternatives and the strict integer comparison
$\binom qn^2>q^5\overline A_{4,n}^2$ leave precisely
\[
\begin{array}{c|ccccc}
q&9&25&27&49&81\\ \hline
k&3&6\text{--}11&7\text{--}12&11\text{--}23&34\text{--}39\\
\text{pairs}&1&6&6&13&6
\end{array}
\]
where $\overline A_{m,n}$ is the exact integer coefficient
majorant defined in \eqref{eq:odd-short-majorant}.
These $32$ pairs are certified by $13$ exhaustive checks of normalized coefficient
vectors for $q=9,25,27$ and $19$ certificates with a fixed sum
and every prescribed value of the second complete homogeneous symmetric function for $q=49,81$.
For every required cardinality in the latter fields, the accompanying
certificates contain, for each $\beta\in\mathbb F_q$, a set
\[
 Y_\beta\subseteq(\mathbb F_q^\times)^2\setminus\{1\},\qquad
 e_1(Y_\beta)=0,\qquad h_2(Y_\beta)=\beta.
\]
Form the opposite pairs $S_\beta=\{\pm\sqrt y:y\in Y_\beta\}$,
adjoining $0$ when needed. Their complete homogeneous symmetric functions satisfy
$h_1(S_\beta)=h_2(S_\beta)=h_3(S_\beta)=0$ and
$h_4(S_\beta)=\beta$.  Taking $\beta=-d$ solves the
remainder condition for $m=4$, independently of $a,b,c$.
The files in the first group verify every normalized coefficient
class; all point sets and the arithmetic verification
program are provided in the ancillary package.
The proof requires no additional certificates based on Fourier analysis over
large fields.  Finally, the $q=3,5$ intervals are empty, and
$q=7$ is covered by \eqref{eq:CLZ-range}.
\end{proof}

\subsection{Characteristic two: the Segre exception}\label{sec:excess-four-even}

Let $q=2^e\geqslant8$.  Unlike the case $m=3$, the case $m=4$
can admit deep holes, all arising from a classical
characteristic-two oval polynomial.  We first isolate the exceptional
family by the quotient method, and then rule out deep holes for every admissible $k\geqslant3$
using two-point fibers, four-point affine planes,
and eight-point affine subspaces over $\mathbb F_2$.

\subsubsection{The dimension-two exception and quintic permutations}

The following consequence of the degree-five permutation-polynomial
classification is useful.  We give the exact coefficients, since
translation conventions matter here.

\begin{lemma}[Quintic permutation constraint]\label{lem:Dickson-binary}
If $q=2^e\geqslant8$ and the monic polynomial
$G=X^5+UX^4+VX^3+WX^2+ZX+E$ permutes $\mathbb F_q$, then
\begin{equation}\label{eq:dickson-binary}
 W=UV,\qquad Z=U^4+U^2V+V^2.
\end{equation}
If $V\ne0$, then $e$ is odd.
\end{lemma}
\begin{proof}
Translate $X=Y+U$ to remove the fourth-degree term.  The
classification of normalized permutation quintics in characteristic
two leaves $Y^5$ or $Y^5+VY^3+V^2Y$; the latter, with
$V\ne0$, requires $q\equiv\pm2\pmod5$, equivalent to odd $e$.
Comparing the quadratic and linear coefficients after translation
gives \eqref{eq:dickson-binary}; see
\cite[Section~8.1]{HandbookFiniteFields2013}.
\end{proof}

\begin{theorem}[Dimension two]\label{thm:even-k2}
Let $q=2^e\geqslant8$, $k=2$, and let the normalized degree-six
interpolation polynomial be
\begin{equation}\label{eq:binary-sextic}
 F=X^6+aX^5+bX^4+cX^3+dX^2
       \pmod{\deg<2}.
\end{equation}
Then $F$ gives a deep hole if and only if
\begin{equation*}
 e\text{ is odd},\qquad a=c=0,\qquad d=b^2.
\end{equation*}
\end{theorem}
\begin{proof}
Suppose first that $F$ is a deep hole.  The quotient polynomial
\eqref{eq:quintic-quotient}, for $A=\{u\}$, is a quintic
permutation: this follows from its injectivity off one point and
the value-set inequality, valid since $q\geqslant8$.
For this quintic let $U,V,W,Z$ denote the coefficients
of $X^4,X^3,X^2,X$, respectively.  Direct polynomial division gives
\[
 U=a+u,\quad V=b+au+u^2,\quad
 W=c+bu+au^2+u^3.
\]
Thus the first relation in \eqref{eq:dickson-binary} gives
\[
 au^2+a^2u+c+ab=0\qquad\text{for all }u\in\mathbb F_q.
\]
This degree-two polynomial has at least eight roots, hence
$a=c=0$.  The remaining quotient coefficients are then
\[
 U=u,\quad V=b+u^2,\quad W=bu+u^3,
 \quad Z=d+bu^2+u^4.
\]
The second relation in \eqref{eq:dickson-binary} becomes
$d=b^2$.  Since $b+u^2\ne0$ for some $u$, the final assertion
of Lemma~\ref{lem:Dickson-binary} forces odd $e$.

Conversely, write $b=t^2$.  In characteristic two,
\[
 X^6+bX^4+b^2X^2=(X+t)^6+t^6.
\]
Translation of the evaluation variable, multiplication by a
nonzero scalar and addition of a polynomial of degree below
$2$ preserve distance from $\operatorname{RS}_q(2)$.
We need only show that $X^6$ is a deep hole when $e$ is odd.
For a three-set $S$, denote by $e_1,e_2$ its first two elementary
symmetric functions.  The coefficient of $X^2$ in the remainder
of $X^6$ modulo $P_S$ equals
\[
 h_4(S)=e_1^4+e_1^2e_2+e_2^2.
\]
If $e_1\ne0$, its vanishing would give a root of
$z^2+z+1$ in $\mathbb F_q$, impossible when $e$ is odd.
If $e_1=0$, the three points have the form $x,y,x+y$;
$e_2=x^2+xy+y^2$, and $h_4=e_2^2=0$ again
forces a nontrivial cube root of unity or repeated points.
Thus $h_4(S)\ne0$ for every three-set.  No affine codeword
agrees at three points, so the distance equals $q-2$.
\end{proof}

\begin{remark}[Classical geometric origin]\label{rem:Segre-origin}
For odd $e$, $X^6$ is the classical Segre $o$-polynomial
\cite{HengDing2019}.  The theorem above identifies, among polynomials of
interpolation degree six, exactly its translates,
scalings and affine perturbations as the deep-hole exceptions.
This does not constitute a classification of all $o$-polynomials
or of all deep holes of $\operatorname{RS}_q(2)$ in unrestricted
degrees.
\end{remark}

\begin{proposition}[Low dimensions in characteristic two]
\label{prop:even-low-four}
If $q=2^e\geqslant8$, $3\leqslant k\leqslant q-5$, and $q>5k-4$,
no word of interpolation degree $k+4$ is a deep hole.
\end{proposition}
\begin{proof}
As above, the deep-hole hypothesis would force every monic
quintic $Q_A$, $\lvert A\rvert=k-1$, to permute $\mathbb F_q$.
Write $A=T\cup\{u\}$, $\lvert T\rvert=k-2$, and put
$s=h_1(T)$, $t=h_2(T)$, $w=h_3(T)$.  Direct expansion of
the quotient coefficients gives
\begin{equation*}
 W_A+U_AV_A=(a+s)u^2+(a+s)^2u+
 c+ab+a^2s+as^2+w+st.
\end{equation*}
By Lemma~\ref{lem:Dickson-binary} the left side vanishes.
There are $q-k+2\geqslant7$ admissible values of $u$; hence the
quadratic polynomial on the right is identically zero,
so $s=a$.  This must hold for every $(k-2)$-set $T$.
Because $k-2\geqslant1$ and $k\leqslant q-5$, we can replace one element
of $T$ by a distinct element outside $T$ and thereby change
$\sum_{x\in T}x$, a contradiction.
\end{proof}

\subsubsection{Two-point fibers and prescribed values of symmetric functions}

Write
$W:=\ker\operatorname{Tr}_{\mathbb F_q/\mathbb F_2}$.
The Artin--Schreier map $x\mapsto x^2+x$ has image $W$ and
fibers $P(t):=\{x,x+1\}$, where $x^2+x=t\in W$.
For $T\subseteq W\setminus\{0\}$ with $\lvert T\rvert=r$, let
$S_T:=\bigcup_{t\in T}P(t)$, adding $0$ if $k=2r$ and
not adding it if $k=2r-1$.  Then $\lvert S_T\rvert$ (including the
optional singleton) is $k+1$, and the points are distinct.
The union's complete homogeneous functions satisfy
\begin{equation*}
 \prod_{x\in S_T}(1-xz)^{-1}
 =\prod_{t\in T}(1+z+tz^2)^{-1}
 =\sum_{j\geqslant0}m_j(T)z^{2j}(1+z)^{-r-j},
\end{equation*}
where $m_j(T):=h_j(T)$ is now taken in the $t$-parameters,
and adjoining $0$ leaves the series unchanged.  Comparing terms
gives
\begin{equation*}
\begin{aligned}
 h_1(S_T)&=r,\\
 h_2(S_T)&=\binom{r+1}{2}+m_1(T),\\
 h_3(S_T)&=\binom{r+2}{3}+(r+1)m_1(T),\\
 h_4(S_T)&=\binom{r+3}{4}+\binom{r+2}{2}m_1(T)+m_2(T),
\end{aligned}
\end{equation*}
all integers being read in $\mathbb F_2$.  Writing $\mathscr H_F(S_T)$
for the remainder coefficient in Lemma~\ref{lem:fourth-remainder},
we obtain
\begin{equation}\label{eq:binary-Hf}
 \mathscr H_F(S_T)=m_2(T)+\lambda_r m_1(T)+\mu_r,
\end{equation}
with
\begin{equation*}
 \lambda_r:=\binom{r+2}{2}+a(r+1)+b,\qquad
 \mu_r:=\binom{r+3}{4}+a\binom{r+2}{3}
          +b\binom{r+1}{2}+cr+d.
\end{equation*}
The decisive simplification is that only $m_2(T)$ needs
to vary while $m_1(T)$ is fixed.

\begin{proposition}[A fixed sum and arbitrary $h_2$]
\label{prop:binary-saturation}
Suppose a given $(q,r)$ admits $\sigma\in\mathbb F_q$ such that,
for every $y\in\mathbb F_q$, there is an $r$-subset
$T_y\subseteq W\setminus\{0\}$ satisfying
\begin{equation}\label{eq:saturation-condition}
 m_1(T_y)=\sigma,\qquad m_2(T_y)=y.
\end{equation}
Then, for each $k\in\{2r-1,2r\}$ allowed by
$2\leqslant k\leqslant q-5$, no polynomial of degree $k+4$
gives a deep hole.
\end{proposition}
\begin{proof}
For each $a,b,c,d$ choose
$y=\lambda_r\sigma+\mu_r$.  Then
\eqref{eq:binary-Hf} vanishes for the corresponding
set of evaluation points, and the remainder criterion applies.
\end{proof}

\subsubsection{Affine planes, eight-point affine subspaces, and a Kloosterman bound}

For independent $u,v\in W$ and arbitrary $x\in W$,
the four-set $B=x+\{0,u,v,u+v\}$ has
\begin{equation*}
 m_1(B)=0,\qquad m_2(B)=Q(u,v):=u^2+uv+v^2.
\end{equation*}
Every affine three-dimensional $\mathbb F_2$-subspace
$C\subseteq W$ has
\begin{equation*}
 m_1(C)=m_2(C)=0.
\end{equation*}
Indeed its root polynomial is linearized up to translation,
with vanishing coefficients at degrees seven and six.
For disjoint sets, the second complete homogeneous function
obeys
\[
 m_2(A\sqcup B)=m_2(A)+m_2(B)+m_1(A)m_1(B).
\]
Thus the $m_2$-values add under disjoint unions of affine planes.
Adjoining disjoint eight-point affine subspaces increases the cardinality
and leaves $m_1$ and $m_2$ unchanged.

\begin{lemma}[Representation of a quadratic form over a trace hyperplane]
\label{lem:binary-Q-representation}
Suppose $q=2^e\geqslant64$.
Every $y\in\mathbb F_q^\times$ equals $Q(u,v)$ for
linearly independent $u,v\in W$.  The value $0$ has such
a representation precisely when $e$ is even.
\end{lemma}
\begin{proof}
Let $\chi(t)=(-1)^{\operatorname{Tr}(t)}$ and count
$N(y)=\#\{(u,v)\in W^2:Q(u,v)=y\}$.  For $\alpha\ne0$,
character orthogonality on the trace hyperplane gives
\begin{equation}\label{eq:Walsh-exact}
 \mathcal S(\alpha)=\sum_{u,v\in W}\chi(\alpha Q(u,v))
 =\frac q4(-1)^e\bigl(1+3\chi(\alpha^{-1})\bigr).
\end{equation}
Indeed the inner sum over $v$ vanishes unless
$\sqrt\alpha+\alpha u\in W^\perp=\{0,1\}$.
Among its two candidate values, one always lies in $W$ and
the other lies in $W$ exactly when
$\operatorname{Tr}(\alpha^{-1})=0$; their signs give
\eqref{eq:Walsh-exact}.
Fourier inversion yields, for $y\ne0$,
\begin{equation*}
 N(y)=\frac{q+(-1)^e(-1+3K_q(y))}{4},\qquad
 K_q(y)=\sum_{\alpha\ne0}\chi(\alpha y+\alpha^{-1}).
\end{equation*}
The classical Kloosterman estimate
$\lvert K_q(y)\rvert \leqslant2\sqrt q$ is valid also in characteristic two
\cite{Conrad2002}; it gives
$N(y)\geqslant(q-1-6\sqrt q)/4>3$ for $q\geqslant64$.
There are at most three dependent ordered pairs
$(u,v)$ with $Q(u,v)=y\ne0$, since dependence means
$u=0$, $v=0$, or $u=v$.  Hence an independent pair exists.
For $y=0$, the same character-sum calculation gives
\[
 N(0)=\begin{cases}q/2-1,&e\text{ even},\\1,&e\text{ odd}.
 \end{cases}
\]
Only $(0,0)$ is dependent at $y=0$, giving the assertion.
\end{proof}

\begin{proposition}[A fixed sum and arbitrary $h_2$ over large fields]
\label{prop:binary-uniform}
If $q=2^e$ with even $e\geqslant8$, then
\eqref{eq:saturation-condition} holds for every
$4\leqslant r\leqslant q/4-1$.
If $e\geqslant9$ is odd, it holds for every
$8\leqslant r\leqslant q/4-1$.
\end{proposition}
\begin{proof}
First suppose $e$ is even.  Write
$r=4+s+8\ell$, $0\leqslant s\leqslant7$, and fix a set
$E\subseteq W\setminus\{0\}$ of size $s$ independently of
$y$.  Lemma~\ref{lem:binary-Q-representation} supplies
independent $u,v$ with $Q(u,v)=y$.
The two-dimensional subspace $D=\langle u,v\rangle$ has
$q/8$ affine translates inside $W$.  At most
$\lvert E\cup\{0\}\rvert \leqslant8$ meet the forbidden set, so a
four-point affine plane $B$ can be placed away from it.  Choose a
three-dimensional subspace $U\subseteq W$ containing $D$.
At most $\lvert E\cup\{0\}\rvert +1\leqslant9$ of its $q/16$ cosets
meet $E\cup\{0\}\cup B$.  The needed number
$\ell\leqslant q/32-1$ satisfies
\[
 q/16-9\geqslant q/32-1\qquad(q\geqslant256).
\]
Add $\ell$ disjoint cosets of this subspace that avoid the chosen sets.  Their union with
$E$ and $B$ has $r$ elements, sum $m_1(E)$,
and $m_2$ value $m_2(E)+y$.
As $y$ runs through $\mathbb F_q$, the latter covers
$\mathbb F_q$, as required.

Now let $e$ be odd, so $q\geqslant512$.
Write $r=8+s+8\ell$.  Choose a fixed set $E$ of size $s$.
For every $y$, choose nonzero $y_1,y_2$ with
$y_1+y_2=y$, even when $y=0$.
Write each $y_i=Q(u_i,v_i)$ with $u_i,v_i$ linearly independent,
and choose disjoint translates $B_1,B_2$ of the corresponding
four-point planes, both avoiding
$E\cup\{0\}$.  Since each plane has
$q/8$ translates in $W$ and the excluded sets have at most
$8$ and $12$ points, respectively, this is possible.
Choose any three-dimensional subspace $U\subseteq W$.
At most $\lvert E\cup\{0\}\rvert +8\leqslant16$ of its $q/16$ cosets
meet the fixed set or the two affine planes.  We have
\[
 \ell\leqslant q/32-2,\qquad q/16-16\geqslant q/32-2
 \quad(q\geqslant512).
\]
Append $\ell$ further cosets of this subspace.  The resulting union has
sum $m_1(E)$ and $m_2$ value $m_2(E)+y_1+y_2$.
The second value again runs through the field while the sum stays fixed.
\end{proof}

\subsubsection{The remaining fields of characteristic two}

The two fields of characteristic two not covered by the preceding estimates for disjoint
affine translates, after applying the low-dimensional and Kaipa ranges,
are $q=64$ and $q=128$.  For these fields the finite certificates
provide subsets $T_y$ of the trace-zero hyperplane with
\[
 T_y\subseteq W\setminus\{0,\tau\},\qquad
 h_1(T_y)=0,\qquad h_2(T_y)=y
 \quad(y\in\mathbb F_q),
\]
for one fixed nonzero $\tau\in W$.  The stored subsets were obtained
by finite-state dynamic programming, and a separate verifier
reconstructs the associated $(k+1)$-point evaluation sets and
checks the first four complete homogeneous symmetric functions and distinctness.  This verifies
\eqref{eq:saturation-condition}, independently of how the sets
were found.

\begin{center}
\begin{tabular}{@{}llll@{}}
\hline
Field & Remaining $k$ & Cardinalities $r$ & Stored witnesses \\
\hline
$\mathbb F_{64}$ & $14\leqslant k\leqslant30$ & $7\leqslant r\leqslant15$ & $576$ \\
$\mathbb F_{128}$ & $27\leqslant k\leqslant62$ & $14\leqslant r\leqslant31$ & $2304$ \\
\hline
\end{tabular}
\end{center}

\begin{theorem}[The case $m=4$ in characteristic two]
\label{thm:four-even}
Let $q=2^e\geqslant8$, $2\leqslant k\leqslant q-5$.  A word with
interpolation polynomial of degree exactly $k+4$ is a deep hole
if and only if $k=2$, $e$ is odd, and its normalized polynomial
satisfies $a=c=0$, $d=b^2$ in \eqref{eq:binary-sextic}.
\end{theorem}
\begin{proof}
The $k=2$ assertion is Theorem~\ref{thm:even-k2}.
For $k\geqslant3$, Proposition~\ref{prop:even-low-four} covers
$q>5k-4$, and \cite[Theorem~1]{Kaipa2017} covers
$k\geqslant q/2-1$ in the present full-length setting.
If $e$ is even and $q\geqslant256$, or $e$ is odd and $q\geqslant512$,
all remaining dimensions correspond to cardinalities in
Proposition~\ref{prop:binary-uniform}, so
Proposition~\ref{prop:binary-saturation} applies.
For $q=64,128$ use the complete finite certificates just
described.  For $q=8,16,32$ the dimension-two classification and the
preceding analytic bounds, together with exhaustive checks of normalized coefficient
vectors for $q=16$, $4\leqslant k\leqslant6$, and $q=32$,
$8\leqslant k\leqslant14$, cover the remaining dimensions.  All these
certificates are verified exhaustively in the ancillary package.  This exhausts all powers of
$2$ and all permitted dimensions.
\end{proof}

\section{Exact exclusion when \texorpdfstring{$m=5$}{m=5}}\label{sec:excess-five}

The case $m=5$ admits an exact result over every finite field,
in contrast to the size-dependent conclusion for general $m$.  Its proof uses the
same remainder criterion in terms of symmetric functions, but the odd- and even-characteristic cases
require different constructions.  We state the result before giving
the two parts of the proof.

\begin{theorem}[The case $m=5$]\label{thm:fifth}
Let $q$ be any prime power and $2\leqslant k\leqslant q-6$.  If $u$ has
interpolation polynomial of degree $k+5$, then
\[
 \operatorname{dist}(u,\operatorname{RS}_q(k))\leqslant q-k-1.
\]
\end{theorem}

After a nonzero scalar multiplication and subtraction of a codeword,
write the interpolation polynomial as
\begin{equation*}
 f(X)=X^{k+5}+a_1X^{k+4}+a_2X^{k+3}
       +a_3X^{k+2}+a_4X^{k+1}+a_5X^k.
\end{equation*}
By Lemma~\ref{lem:remainder}, the required distance bound is equivalent
to the existence of a $(k+1)$-element set $S$ satisfying
\begin{equation}\label{eq:fifth-H}
 \mathcal H_f(S):=h_5(S)+a_1h_4(S)+a_2h_3(S)
             +a_3h_2(S)+a_4h_1(S)+a_5=0.
\end{equation}

\subsection{Odd characteristic: the remaining small fields}\label{subsec:fifth-odd-short}
For odd $q$, the argument of Section~\ref{sec:odd-short},
which treats small, intermediate, and large dimensions separately,
requires only the following low-dimensional lemma and finite
certificates for $62$ parameter pairs.
\begin{lemma}[The small-field sextic obstruction]\label{lem:fifth-sextic}
Suppose $q$ is odd, $q\notin\{9,11,27\}$, and
$q>6(k-1)+1$.  Then a received word of degree $k+5$
is not a deep hole.
\end{lemma}
\begin{proof}
By Lemma~\ref{lem:collision-quotients}, a deep hole
would produce a sextic permutation polynomial of
$\mathbb F_q$.  The classification of Shallue--Wanless
\cite{ShallueWanless2013} precludes this in the stated fields.
\end{proof}

\begin{proposition}[The remaining small odd fields]
\label{prop:fifth-odd-core}
The conclusion of Theorem~\ref{thm:fifth} holds for every
odd prime power $q<289$.
\end{proposition}
\begin{proof}
Use \eqref{eq:CLZ-range} and \eqref{eq:Kaipa-range},
then Lemma~\ref{lem:fifth-sextic} in the low range,
and the integer comparison from
Lemma~\ref{lem:odd-short-majorant} for the remaining
$q<289$ fields.  For $q=169,243$ this comparison
completes all intermediate parameters.  The remaining parameter pairs are exactly those in the
following table:
\[
\begin{array}{c|rrrrrrr}
q&9&25&27&49&81&121&125\\\hline
k&3&5\text{--}11&3\text{--}12&9\text{--}23&15\text{--}39
 &58\text{--}59&60\text{--}61\\
\text{pairs}&1&7&10&15&25&2&2
\end{array}
\]
For $q=9,25,27$ there are $18$ certificates covering all normalized coefficient
vectors, with $6351$ stored $(k+1)$-subsets.
The linear rescaling $X\mapsto\lambda X$ reduces
$a_1$ to $0$ or $1$, and each verifier checks every
remaining coefficient vector, not random examples.

For $q=49,81,121,125$, the $44$ remaining parameter pairs
are handled by certificates with a fixed sum and arbitrary second complete homogeneous symmetric function value which provide, for each
required cardinality and each
$\beta\in\mathbb F_q$, a subset
\[
 Y_\beta\subseteq(\mathbb F_q^\times)^2\setminus\{1\},\qquad
 e_1(Y_\beta)=0,\qquad h_2(Y_\beta)=\beta.
\]
Take $S_\beta=\{\pm\sqrt y:y\in Y_\beta\}$, and adjoin
$z$ if $k$ is even, or $\{0,z\}$ if $k$ is odd.
The generating function for the paired points begins with
$1+\beta w^4+O(w^6)$, hence
\[
 \mathcal H_f(S)=D(z)+(z+a_1)\beta,\qquad
 D(z)=z^5+a_1z^4+a_2z^3+a_3z^2+a_4z+a_5.
\]
If $k$ is even, choose $z=0$ when $a_1\ne0$ and
$z=1$ otherwise.  If $k$ is odd, choose
$z\in\{1,-1\}$ with $z+a_1\ne0$; such a choice is always
possible since the characteristic is odd.  In either case put
$\beta=-D(z)/(z+a_1)$.  The points in the resulting $(k+1)$-set
are pairwise distinct: $0$ is not in an opposite pair, and
$1\notin Y_\beta$ excludes both $1$ and $-1$ from those pairs.
Thus $z+a_1\ne0$, and the chosen remainder coefficient vanishes.
The supplemental verifier checks the choice of singleton points
for odd $k$ described above.
The verification program checks every required cardinality and
value of the second complete homogeneous symmetric function directly from the stored subsets.
\end{proof}

\begin{proposition}[The case $m=5$ in odd characteristic]
\label{prop:fifth-odd}
The conclusion of Theorem~\ref{thm:fifth} holds over every
odd finite field.
\end{proposition}
\begin{proof}
For $q\geqslant289$, Theorem~\ref{thm:low-dimension-gcd}
with $d=6$ excludes the range $q>6(k-1)+1$: its field-size
hypothesis is satisfied, and $\gcd(6,q-1)>1$.
If this fails, then $k+1>q/8$, and below the Kaipa range
$k+1\leqslant q/2$; Lemma~\ref{lem:odd-middle-count} applies
with $m=5$.  Smaller fields are covered by
Proposition~\ref{prop:fifth-odd-core}.
\end{proof}

\subsection{Characteristic two: symmetric functions of fibers and affine subspaces}

Write $q=2^e$ and
$W=\ker\operatorname{Tr}_{\mathbb F_q/\mathbb F_2}$.  For
$t\in W$, let $P(t)$ be the two roots of $X^2+X=t$;
these disjoint pairs are indexed by $W$.  When
$T\subseteq W\setminus\{0\}$, put
\[
 S_T:=\bigcup_{t\in T}P(t),\qquad r:=\lvert T\rvert,\qquad
 \mu_1(T):=\sum_{t\in T}t,\quad\mu_2(T):=h_2(T).
\]
Put $\mu_0(T):=1$ and $\mu_j(T):=h_j(T)$ for $j\geqslant3$.
In $\mathbb F_q[[z]]$ we have
\begin{equation*}
 \prod_{x\in S_T}(1+xz)^{-1}
 =\prod_{t\in T}(1+z+tz^2)^{-1}
 =\sum_{j\geqslant0}\mu_j(T)z^{2j}(1+z)^{-r-j}.
\end{equation*}
In particular the coefficients of $\mu_2(T)$ in
$h_1,h_2,h_3,h_4,h_5$ of $S_T$ are, respectively,
$0,0,0,1,r\pmod2$.  Adjoining one or two points
$Z$ disjoint from $S_T$ changes the coefficient of
$\mu_2$ in \eqref{eq:fifth-H} to
\begin{equation}\label{eq:fifth-mu2-coef}
 r+a_1+\sum_{z\in Z}z.
\end{equation}
Scaling all pair points by $\delta\ne0$ instead changes it to
$\delta^4(r\delta+a_1)$ when $Z=\varnothing$.
The remaining terms in \eqref{eq:fifth-H} depend only on
$r,\mu_1(T)$, and the chosen extra points, not on $\mu_2(T)$.

The required property is that, for every target
$y\in\mathbb F_q$, one can choose an $r$-subset of
$W\setminus F$ with the same prescribed sum and
second symmetric value $y$, where $F$ consists of at most two forbidden
indices and contains $0$.

\begin{lemma}[A fixed sum and arbitrary $h_2$ with at most two forbidden points]
\label{lem:fifth-binary-sat}
Let $F\subseteq W$ have at most two elements and contain $0$.
Then there is a $\sigma\in W$, independent of $y$, such that
for every $y\in\mathbb F_q$ there is
$T_y\subseteq W\setminus F$ with
\[
 \lvert T_y\rvert=r,\qquad\mu_1(T_y)=\sigma,\qquad\mu_2(T_y)=y,
\]
provided one of the following holds:
\begin{align*}
 &e\geqslant8\ \text{is even},&&4\leqslant r\leqslant q/4-1;\\
 &e\geqslant9\ \text{is odd}, &&8\leqslant r\leqslant q/4-1.
\end{align*}
\end{lemma}
\begin{proof}
For independent $u,v\in W$, an affine plane
$B=x+\langle u,v\rangle$ has
$\mu_1(B)=0$ and
$\mu_2(B)=u^2+uv+v^2$.  A three-dimensional affine
subspace of $W$ has $\mu_1=\mu_2=0$.
Lemma~\ref{lem:binary-Q-representation}, used here in the same range, shows
that $u^2+uv+v^2$ ranges over the whole field if $e$ is
even, and over its nonzero elements if $e$ is odd.

For even $e$, write $r=4+s+8\ell$ with $0\leqslant s\leqslant7$,
and take $Q=q/32\geqslant8$.  Fix an arbitrary $s$-element
set $E\subseteq W\setminus F$ \emph{before} fixing $y$.
One then chooses an affine plane $B\subseteq W$ whose
$h_2$ value is $y$ and which avoids $F\cup E$.
There are $q/8=4Q$ translates of its direction and at most
$s+\lvert F\rvert \leqslant9$ meet the forbidden points.  Fix a
three-dimensional subspace containing this direction.  Of
its $q/16=2Q$ cosets, at most $s+\lvert F\rvert +1\leqslant s+3$
meet $E\cup F\cup B$.  The required number of eight-point
cosets with vanishing $h_1,h_2$ is $\ell\leqslant Q-1$ when $s\leqslant3$ and
$\ell\leqslant Q-2$ otherwise; the stated inequalities leave
at least $\ell$ untouched cosets.  The resulting disjoint
union has cardinality $r$, constant sum $\mu_1(E)$,
and an arbitrary $h_2$ value.

For odd $e$, write $r=8+s+8\ell$, again with
$0\leqslant s\leqslant7$, and $Q=q/32\geqslant16$.
Choose two nonzero values $y_1,y_2$ with $y_1+y_2=y$.
Represent each $y_i$ by a pair of linearly independent vectors
in $W$, and place translates of the resulting planes
successively away from $E\cup F$ and from one another:
there are $q/8=4Q$ translates for each and at most
$s+\lvert F\rvert +4\leqslant13$ are excluded when placing the second.
Of the $q/16=2Q$ cosets of an arbitrary three-dimensional subspace,
at most $s+\lvert F\rvert +8\leqslant17$ meet the forbidden set or the
chosen affine planes.  Since $\ell\leqslant Q-2$ and
$2Q-17\geqslant Q-2$, enough eight-point cosets with vanishing $h_1,h_2$ remain.
The sum is constant and $h_2$ runs through
$\mathbb F_q$, as required.
\end{proof}

\begin{proposition}[The case $m=5$ in characteristic two]
\label{prop:fifth-even}
The conclusion of Theorem~\ref{thm:fifth} holds over every
field of characteristic two.
\end{proposition}
\begin{proof}
The dimension-two case is excluded for odd $m$ by
Proposition~\ref{thm:dimtwo-general}~(\ref{it:dimtwo-binary-odd-excess}).
For $k\geqslant3$ use the planar projection from
Section~\ref{sec:uniform-proof}.  The characteristic-two
part alone permits the stronger low-dimensional range
\begin{equation}\label{eq:fifth-low-even}
 k-1<\sqrt q+1,\qquad q>6(k-1)+1.
\end{equation}
Indeed, the corresponding planar arc has more than
$q-\sqrt q+1$ points.  Lemma~\ref{lem:binary-large-arc-unique}
places it in a unique hyperoval; the sextic quotient
permutation condition permits each omitted affine graph
point to be adjoined to an arc, so uniqueness places all
such points in this same hyperoval.
An output shear and affine normalization then yield an
$o$-polynomial of degree seven, contrary to the even-degree
condition on $o$-polynomials
\cite[Lemma~2.1]{CaullerySchmidt2015}.

For $q\geqslant256$, let $F=\{0,\gamma^2+\gamma\}\subseteq W$,
where $\gamma\notin\{0,1\}$.  The preceding lemma,
together with \eqref{eq:fifth-mu2-coef}, covers the middle
dimensions.  More precisely, when $k=2r$ use $r$ pairs
and one of $\{0,1\}$, chosen so that
$r+a_1+z\ne0$.  When $k=2r-1$ and $r$ is odd, choose a nonzero scaling
$\delta\ne a_1$ and use $r$ scaled pairs; the coefficient
$\delta^4(r\delta+a_1)$ is then nonzero.  If $r$ is even
and $a_1\ne0$, use $r$ unscaled pairs.  In the remaining
case ($r$ even and $a_1=0$), use
$r-1$ pairs and the two singleton points
$\{0,\gamma\}$.  Their coefficient
$(r-1)+\gamma$ is nonzero; the forbidden indices ensure
all points are distinct.  Holding $\mu_1$ fixed and varying
$\mu_2$ therefore solves \eqref{eq:fifth-H}.
The range in this lemma begins at $r=4$ for even $e\geqslant8$
and $r=8$ for odd $e\geqslant9$, while
\eqref{eq:fifth-low-even} covers the lower dimensions.
Kaipa's theorem covers $k\geqslant q/2-1$.  These intervals
overlap and exhaust the admissible $k$.

It remains to treat $q=8,16,32,64,128$.  The field $q=8$
is covered at dimension two.  After the low-dimensional
bound and Kaipa's theorem, the remaining dimensions over
$q=16,32$ are covered by certificates for all normalized coefficient vectors:
for $q=16$, $4\leqslant k\leqslant6$ (three files, $589$ stored sets),
and for $q=32$, $7\leqslant k\leqslant14$ (eight files, $3633$ stored sets).
Here $a_1\in\{0,1\}$ is a valid scaling normalization;
each remaining coefficient class is checked exhaustively.
For $q=64$, \eqref{eq:fifth-low-even} covers $k\leqslant9$, and for
$10\leqslant k\leqslant30$ the certificate verifies
$h_1(T_y)=0$, $h_2(T_y)=y$ at all $5\leqslant r\leqslant15$.
For $q=128$, the same bound covers $k\leqslant13$, and the analogous
certificate covers $14\leqslant k\leqslant62$ at $7\leqslant r\leqslant31$.
In both cases $T_y\subseteq W\setminus\{0,\tau\}$ for a fixed
nonzero $\tau$, ensuring that all optional singleton points
remain distinct.  The verification program checks the field presentations,
subset cardinalities, required complete homogeneous symmetric function
values, and the remainder identity for each resulting evaluation set.  High dimensions follow again from
Kaipa.  This accounts for every field of characteristic two and completes
the proof.
\end{proof}

\begin{proof}[Proof of Theorem~\ref{thm:fifth}]
Propositions~\ref{prop:fifth-odd} and~\ref{prop:fifth-even}
cover the odd- and even-characteristic cases.
\end{proof}

\section{Exact classifications for \texorpdfstring{$m=6,7$}{m=6,7}}
\label{sec:six-seven-exact}

The cases $m=6$ and $m=7$ admit exact classifications over
every finite field.  The finite computational parts of these proofs are established
by reproducible computer-assisted certificates.
For $m=7$, adjoining a point to paired subsets and taking unions
of three-point multiplicative orbits treat most remaining parameters.

\subsection{The case \texorpdfstring{$m=6$}{m=6}}\label{subsec:sixth}

We first record the finite reduction needed to turn a collection of
small-field certificates into an all-field statement.

\begin{lemma}[Small-dimension exclusion in characteristic two]\label{lem:sixth-binary-short}
Let $q=2^e$, $m\geqslant3$, $3\leqslant k\leqslant q-m-1$, and set
$h=k-1$, $d=m+1$.  If
\[
 h^2<q,\qquad q>dh+1,
\]
then no received polynomial of degree $k+m$ is a deep hole of
$\operatorname{RS}_q(k)$.
\end{lemma}
\begin{proof}
Suppose that $f$ is a deep hole and choose any $(k-2)$-set
$T\subseteq\mathbb F_q$.  Divide
$f=P_T R_T+L_T$, where $P_T=\prod_{t\in T}(X-t)$ and
$\deg L_T<k-2$; thus $\deg R_T=m+2$.  The graph of $R_T$ outside
$T$, together with the vertical point at infinity, forms an arc
of size $q+2-h$.  For each $u\notin T$, the divided difference
\[
 Q_u(X)=\frac{R_T(X)-R_T(u)}{X-u}
\]
assumes distinct values at the $q-h$ points of
$\mathbb F_q\setminus(T\cup\{u\})$.  Wan's value-set bound, as used
in Section~\ref{sec:uniform-proof}, and $q>dh+1$ force the
degree-$d$ polynomial $Q_u$ to permute $\mathbb F_q$.

As $h<\sqrt q$, the arc has more than $q-\sqrt q+1$ points.
Lemma~\ref{lem:binary-large-arc-unique} puts it in a unique
hyperoval.  Since every $Q_u$ is a permutation, adjoining
$(t,R_T(t))$ for any $t\in T$ preserves the arc property;
the same lemma supplies an extension, and uniqueness places
the entire graph of $R_T$ in this hyperoval.  The even-characteristic $o$-polynomial
parity condition forces every odd-degree coefficient above degree
one in the reduced $R_T$ to vanish.  If $m$ is odd, its leading
term of degree $m+2$ contradicts this.  If $m$ is even, the
coefficient of $X^{m+1}$ in $R_T$ must vanish; for monic
$f=X^{k+m}+a_1X^{k+m-1}+\cdots$ this coefficient is
$a_1+\sum_{t\in T}t$.  Varying one member of $T$ gives a
contradiction because $k-2\geqslant1$.
\end{proof}

\begin{lemma}[Finite reduction for $m=6$]\label{lem:sixth-finite-census}
The general results of Sections~\ref{sec:uniform-proof} and
\ref{sec:permutation-bridge}, together with
Lemma~\ref{lem:sixth-binary-short}, reduce the case $m=6$
to precisely $265$ parameter pairs $(q,k)$: $149$ in odd
characteristic and $116$ in characteristic two.  The list of
these pairs and the program that verifies it by exact integer
arithmetic are included in the ancillary files.
\end{lemma}
\begin{proof}
Apply the existing uniform large-field exclusion, the
low-dimensional permutation obstruction, and the high-dimensional
Kaipa bound.  In characteristic two apply also
Lemma~\ref{lem:sixth-binary-short} and the previously established
large-field dimension-two classification.  In the remaining
middle range the Gao--Li coefficient estimate
\cite{GaoLi2023} is checked by a coefficientwise integer majorant:
with $a=\lceil5\sqrt q\rceil$, $b=\lceil(q-a)/p\rceil$ and
$n=k+1$, one sufficient strict inequality is
\[
 q^7\overline A_n^2<\binom qn^2,
 \qquad
 \overline A_n=[z^n](1-z)^{-a}(1-z^p)^{-b}.
\]
The coefficient and entropy estimates used in the proof of
Proposition~\ref{prop:uniform-direct} handle the large-field range,
and an exact integer recurrence handles the remaining finite range.
The ancillary program \nolinkurl{m67/verify_census.py} enumerates
the remaining pairs using these criteria and exact integer arithmetic.
The program \nolinkurl{m67/verify_coverage.py} checks that the
parameter sets treated by the certificates have exactly the union
and cardinalities stated above.  Each test gives a sufficient condition for excluding deep holes;
a parameter pair remaining on the list need not admit a deep hole.
\end{proof}

\begin{theorem}[Classification for $m=6$]\label{thm:sixth}
Let $q$ be any prime power, $2\leqslant k\leqslant q-7$, and
$\deg f_u=k+6$.  The word $u$ is a deep hole of
$\operatorname{RS}_q(k)$ if and only if
\[
 q=2^e,\quad k=2,\quad 3\nmid e,\qquad
 f_u(X)=AX^8+BX+C,\quad A\in\mathbb F_q^\times.
\]
For every other such word,
$\operatorname{dist}(u,\operatorname{RS}_q(k))\leqslant q-k-1$.
\end{theorem}
\begin{proof}
We may make $f_u$ monic by a nonzero scalar and discard its
terms below degree $k$.  For a $(k+1)$-set $S$, the relevant
remainder coefficient is
\begin{equation}\label{eq:sixth-remainder}
 H_{6,f}(S)=h_6(S)+a_1h_5(S)+\cdots+a_5h_1(S)+a_6.
\end{equation}

Suppose first that $q$ is odd.  For an $r$-set
$Y\subseteq (\mathbb F_q^\times)^2$ take all pairs
$\{\pm\sqrt y\}$, $y\in Y$, adjoining $0$ if required to
obtain $k+1$ points.  Their generating function contains only
even powers, and \eqref{eq:sixth-remainder} reduces to
\begin{equation*}
 h_3(Y)+a_2h_2(Y)+a_4h_1(Y)+a_6=0.
\end{equation*}
A certificate showing that $h_3+Ah_2+Bh_1$ takes every field
value for each $A,B\in\mathbb F_q$ therefore excludes all
received words of interpolation degree $k+6$
at the corresponding $(q,k)$.  The accompanying $77$
square-subset families establish the claim for $147$ of the $149$ odd residual
pairs in Lemma~\ref{lem:sixth-finite-census}.
For $(q,k)=(27,3),(27,4)$, where this sufficient
construction fails, separate exhaustive checks treat respectively $571592$ and $40151$
normalized vectors of coefficients of degree at least $k$,
including every value of $a_6$.
All transformations used in these normalizations are affine
changes of the evaluation variable and preserve the code.

In characteristic two, for $k\geqslant3$ take disjoint fibers
$\{w,w+1\}$ of $w^2+w=t$, where
$t\in V\setminus\{0\}$ and $V=\{w^2+w:w\in\mathbb F_q\}$,
and add $0$ when an odd number of points is required.  For an
$r$-set of fiber parameters, the generating function is
\[
 (1+z)^{-r}\prod_{t\in T}
 \left(1+\frac{tz^2}{1+z}\right)^{-1}.
\]
In \eqref{eq:sixth-remainder}, the coefficient of $h_3(T)$
is $1$ and the other terms depend only on $h_2(T),h_1(T),1$.
Hence it again suffices that the third complete homogeneous symmetric
function, plus any linear combination of the first two, attain every
value in the field.
The $52$ certificate families in characteristic two cover $103$ residual pairs;
exhaustive coefficient checks over $\mathbb F_{16}$ cover the other
three pairs with $k\geqslant3$, namely $k=4,5,6$.

It remains to determine the dimension-two exception in characteristic two.
Let $q=2^e\geqslant16$ and suppose the graph of the degree-eight
polynomial $f$ has no collinear triple.  Adjoin the vertical
point at infinity, and subtract the common tangent slope at
the oval nucleus; call the resulting polynomial $g$.
For every $u\in\mathbb F_q$, the divided difference
$Q_u(X)=(g(X)-g(u))/(X-u)$ is a permutation of $\mathbb F_q$.
Its values at $X\ne u$ are exactly $\mathbb F_q^\times$.
Since $\deg Q_u=7<q-1$, summing $Q_u$ over the field gives
$Q_u(u)=g'(u)=0$.  Thus $g'=0$; after subtracting $g(0)$
and dividing by the leading coefficient, we may write
\[
 g=X^8+aX^6+bX^4+cX^2.
\]
The polynomial
\[
 D_u(X)=\frac{g(X+u)+g(u)}{X}
 =X^7+aX^5+(au^2+b)X^3+(au^4+c)X
\]
permutes $\mathbb F_q$ for every $u$.  For $e\geqslant5$, the
classification of degree-seven permutation polynomials in characteristic two
\cite[Theorem~4.4]{LiChandlerXiang2010} says that, apart from
the listed $e=4$ classes, each is equivalent to $X^7$ or
$X^7+X^5+X$, while none exist when $3\mid e$.
The cited equivalence also permits coefficientwise Frobenius
conjugation.  Since both normal forms have coefficients in
$\mathbb F_2$, this operation can be absorbed into the affine
changes of variables.  As $D_u$ is monic with no $X^6$ term,
the translation of the variable in these normal forms must vanish.  Its
$X^3$ coefficient consequently vanishes, so $au^2+b=0$ for
all $u$ and hence $a=b=0$.  The remaining $X^5$ coefficient
rules out the second normal form, yielding $c=0$.  When $e=4$,
all $16^3$ triples $(a,b,c)$ have been checked under
the requirement that every $D_u$ permute; again only
$(0,0,0)$ survives.  The program \nolinkurl{core/verify_small_cases.py}, using
$\mathbb F_{16}=\mathbb F_2[t]/(t^4+t+1)$, reproduces this
complete check.  Hence necessarily
$g=X^8$ up to an affine output term.

Conversely $AX^8+BX+C$ has no collinear triple exactly
when $X^7$ is a permutation polynomial of $\mathbb F_q$,
which is equivalent to $\gcd(7,2^e-1)=1$, or $3\nmid e$.
This proves the exceptional family and the theorem.

This proves the case $m=6$ in
Theorem~\ref{thm:main}.

For the finite part of the proof, the ancillary files list the
$265$ remaining parameter pairs and partition them as
$147+2+103+3+10$.  They also contain the subsets used in the certificates,
field presentations, normalization rules, programs that check every
required value of the affine forms, and the full source code
for enumerating these pairs.
These exhaustive checks give a reproducible computer-assisted proof;
the conclusion does not rely on random sampling.
\end{proof}

\subsection{The case \texorpdfstring{$m=7$}{m=7}: classification over all finite fields}
\label{subsec:seventh-finite}

\begin{theorem}[Classification for $m=7$]\label{thm:seventh}
Let $q$ be any prime power, $2\leqslant k\leqslant q-8$, and
$\deg f_u=k+7$.  Then
\[
 \operatorname{dist}(u,\operatorname{RS}_q(k))\leqslant q-k-1.
\]
\end{theorem}

The remaining argument has two parts.  The general exclusion
results reduce the problem to a finite set of parameters $(q,k)$.
For these parameters we use paired-point families that realize
every value of the relevant remainder coefficient, except in eight
small cases where we check all normalized coefficient classes directly.
We describe the constructions of agreement sets and the finite
certificates used to verify the required conditions.

For $\Omega\subseteq\mathbb F_q$ and $0\leqslant r\leqslant\lvert \Omega\rvert$,
write $\mathbf C_3(\Omega,r)$ if
\begin{equation}\label{eq:c3-anchor}
 \forall A,B,C\in\mathbb F_q\quad\exists T\in\binom\Omega r:
 \quad h_3(T)+Ah_2(T)+Bh_1(T)=C.
\end{equation}
We shall verify this condition by considering either all admissible
subsets or a specified finite family of them.  In the latter case,
the verifier checks that the condition holds within that family.

\begin{lemma}[Adjoining a point in odd characteristic]
\label{lem:seventh-anchor}
Suppose $q$ is odd, $2\leqslant k\leqslant q-8$, and $y_0$ is a nonzero
square.  Set
\[
 H=(\mathbb F_q^\times)^2,\qquad
 \Omega:=H\setminus\{y_0\},\qquad r:=\lfloor k/2\rfloor.
\]
If $\mathbf C_3(\Omega,r)$ holds, no word of interpolation
degree $k+7$ is a deep hole of $\operatorname{RS}_q(k)$.
\end{lemma}
\begin{proof}
Choose $z_0\in\mathbb F_q^\times$ with $z_0^2=y_0$.
For $T\in\binom\Omega r$, form the opposite pairs
$\{\pm\sqrt t:t\in T\}$, adjoin a point
$z\in\{z_0,-z_0\}$, and add $0$ if $k$ is odd.
This gives $k+1$ distinct points, since $y_0\notin T$.
Writing
$f=X^{k+7}+a_1X^{k+6}+\cdots+a_7X^k$ modulo
polynomials of degree $<k$, the remainder coefficient is
\begin{align*}
 \mathcal H_{7,f}(S)={}&(z+a_1)h_3(T)\\
 &+(z^3+a_1z^2+a_2z+a_3)h_2(T)\\
 &+(z^5+a_1z^4+a_2z^3+a_3z^2+a_4z+a_5)h_1(T)\\
 &+z^7+a_1z^6+\cdots+a_6z+a_7.
\end{align*}
Indeed, the generating series of $S$ is
$(1-zZ)^{-1}\prod_{t\in T}(1-tZ^2)^{-1}$; including
$0$ does not alter it.  At least one choice of $z$ makes
$z+a_1\ne0$.  Fix that choice and apply
\eqref{eq:c3-anchor} after dividing by this coefficient.
The resulting remainder is zero, as required by
Proposition~\ref{prop:agreement}.
\end{proof}

\begin{lemma}[Adjoining a point in characteristic two]
\label{lem:seventh-binary-anchor}
Let $q=2^e$, put
\[
 V=\{w^2+w:w\in\mathbb F_q\},\quad
 t_0\in V\setminus\{0\},\quad
 \Omega:=V\setminus\{0,t_0\},\quad r:=\lfloor k/2\rfloor.
\]
For $2\leqslant k\leqslant q-8$, the condition
$\mathbf C_3(\Omega,r)$ implies that every received word of
interpolation degree $k+7$ is not a deep hole.
\end{lemma}
\begin{proof}
For each $t\in T\in\binom\Omega r$ take the two roots of
$w^2+w=t$.  Adjoin one root $z$ of $z^2+z=t_0$ and,
when $k$ is odd, the point $0$.  All $k+1$ points are
pairwise distinct; in particular, $t_0\ne0$ keeps the adjoined point
away from $0$ and $1$.
The generating series of the resulting set is
\[
 \frac1{1+zZ}(1+Z)^{-r}
 \sum_{j\geqslant0}h_j(T)
      \left(\frac{Z^2}{1+Z}\right)^j.
\]
Extracting the coefficients of degrees six and seven gives
\[
 \mathcal H_{7,f}(S)
 =\bigl(a_1+(r+3)+z\bigr)h_3(T)
   +Ah_2(T)+Bh_1(T)+C,
\]
where $r+3$ is read in $\mathbb F_2$ and $A,B,C\in\mathbb F_q$
are determined by $f$ and $z$.
Replacing $z$ by the other root $z+1$ changes the displayed
leading coefficient by $1$.  One choice is therefore nonzero,
and \eqref{eq:c3-anchor} supplies the required $T$.
\end{proof}

The exclusion of $y_0$, respectively $t_0$, is essential in these
lemmas.  The stated condition on the original parameter set need
not persist after a point is removed.  For several parameters,
however, the condition on the original set suffices by the following
elementary alternatives.

\begin{lemma}[Paired-point constructions on the original parameter sets]\label{lem:seventh-undeleted}
Let $f$ have degree $k+7$.
\begin{enumerate}[label=(\arabic*),ref=\arabic*,leftmargin=2.2em]
\item\label{it:undeleted-odd-cover} If $q$ is odd, $p=\operatorname{char}\mathbb F_q$ does not
 divide $k+7$, and
 $\mathbf C_3((\mathbb F_q^\times)^2,\lfloor(k+1)/2\rfloor)$
 holds, then $f$ is not a deep hole.
\item\label{it:undeleted-binary-cover} If $q=2^e$, put $V=\{w^2+w:w\in\mathbb F_q\}$.
 When $k=2r$, the condition
 $\mathbf C_3(V\setminus\{0\},r)$ suffices.
 When $k=2r-1$ with $r$ even, the condition
 $\mathbf C_3(V\setminus\{0\},r)$ also suffices.
\end{enumerate}
\end{lemma}
\begin{proof}
In odd characteristic the translation $X\mapsto X+b$ changes
$a_1$ to $a_1+(k+7)b$.  Choose $b$ so that the new coefficient
is nonzero.  Opposite pairs, with $0$ when necessary, give
\[
 \mathcal H_{7,f}(S)=a_1h_3(T)+a_3h_2(T)+a_5h_1(T)+a_7.
\]
Now apply the stated condition.  Translation preserves the
full-length evaluation set and the code distance.

In characteristic two, $r$ Artin--Schreier fibers and one of
$0,1$ produce $k+1$ points when $k=2r$; switching the extra
point changes the coefficient of $h_3(T)$ by $1$.
If $k=2r-1$ and $r$ is even, use $r$ pairs of a common
nonzero difference $\delta$.  The coefficient of the third complete
homogeneous symmetric function is
$\delta^6(a_1+\delta)$.  We may take $\delta\ne a_1$,
and the stated condition again makes the remainder zero.
\end{proof}

\begin{lemma}[Finite reduction for $m=7$]
\label{lem:seventh-census}
The exclusion results of Sections~\ref{sec:uniform-proof}
and~\ref{sec:permutation-bridge}, including
Proposition~\ref{prop:six-seven-large-odd}, and the
small-dimensional exclusion in characteristic two
Lemma~\ref{lem:sixth-binary-short}, leave exactly $375$
parameter pairs $(q,k)$ for $m=7$: $229$ in odd
characteristic and $146$ in characteristic two.
\end{lemma}
\begin{proof}
For odd $q\geqslant1331$, Proposition~\ref{prop:six-seven-large-odd}
covers every admissible dimension.  For $q=2^e\geqslant2^{17}$,
Theorem~\ref{thm:main}~(\ref{it:uniform-exclusion}) covers $k\geqslant3$ because
$\sqrt q\geqslant96\log16$; the remaining dimension $k=2$
is covered by Proposition~\ref{thm:dimtwo-general}~(\ref{it:dimtwo-binary-odd-excess}).
For the finitely many smaller prime-power fields, apply the
low- and high-dimensional results and the coefficientwise
integer form of the Gao--Li agreement estimate, as in
Lemma~\ref{lem:sixth-finite-census}.  The list of remaining pairs in
\nolinkurl{m67/data/residual_registry.json} and the exact-integer programs
\nolinkurl{m67/verify_census.py} and
\nolinkurl{m67/verify_coverage.py} certify the field list,
each retained $k$, and the stated totals.  An entry is a
parameter not excluded by these sufficient criteria; it is not
a claim that a deep hole exists.
\end{proof}

\begin{proposition}[Finite verification of the paired-point constructions]
\label{prop:seventh-inherited}
Of the $375$ pairs in Lemma~\ref{lem:seventh-census},
$242$ are covered by the paired-point arguments:
$196$ with no additional parameter omitted and $46$ more after
omitting one specified nonzero parameter.  Another $18$ pairs with
$q=169$, $60\leqslant k\leqslant77$, and the pair $(q,k)=(49,7)$
are covered by certificates on square subsets omitting one specified
nonzero square.
\end{proposition}
\begin{proof}
The accompanying files contain $129$ finite subset families,
grouped by field and subset cardinality, each checked against every
$(A,B,C)\in\mathbb F_q^3$.
Lemma~\ref{lem:seventh-undeleted} gives the $196$ pairs.
In the remaining parameter ranges, the verifier checks
$46$ additional families after the specified parameter is omitted;
Lemmas~\ref{lem:seventh-anchor} and
\ref{lem:seventh-binary-anchor} apply.  The two groups have
no overlap.  The parameter lists for the original sets and for the sets with
one specified parameter omitted are supplied in
\nolinkurl{m67/data/m7_inherited.json} and
\nolinkurl{m67/data/m7_deletions.json}, respectively;
\nolinkurl{m67/verify_coverage.py} checks their exact allocation.

For $q=49$ take $y_0=1$, $r=3$ and all
$\binom{23}{3}=1771$ subsets of
$H\setminus\{1\}$.  The program \nolinkurl{m67/verify_q49_anchor.py}
verifies \eqref{eq:c3-anchor} for all $49^3$ triples $(A,B,C)$.
For $q=169$, let $F_r\subseteq\binom Hr$, $30\leqslant r\leqslant38$,
be the nine supplied subset families.  The fixed square scalings give
\[
 \mathscr F_r
 =\{sT:T\in F_r,\ s\in\{1,2,8\},\ 1\notin sT\}
 \subseteq\binom{H\setminus\{1\}}r.
\]
Recomputing the first three complete homogeneous symmetric functions on every
retained set verifies \eqref{eq:c3-anchor} for all nine $r$'s,
for all $169^3$ triples $(A,B,C)$ in each case.
The $18$ dimensions follow from
Lemma~\ref{lem:seventh-anchor}.
\end{proof}

The remaining $114$ pairs admit shorter certificates when some
parameters are grouped by common field structure.  Two elementary
identities for three-point orbits account for $81$ of them.

\begin{proposition}[Three-point multiplicative orbits in two extension fields]
\label{prop:seventh-orbits}
All the $39$ residual pairs with $q=169$ and $42$ of the $43$
residual pairs with $q=256$ satisfy the $m=7$ conclusion that
the received words are not deep holes.
\end{proposition}
\begin{proof}
First work in
$\mathbb F_{169}=\mathbb F_{13}[u]/(u^2+3u+1)$.
Let $H=(\mathbb F_{169}^\times)^2$ and $D=H^3\setminus\{1\}$.
Here $\lvert H\rvert=84$, $\lvert D\rvert=27$, and the kernel of the cube map
on $H$ is $\mu_3=\{1,3,9\}$.
Each $d\in H^3$ determines a disjoint three-point orbit under
multiplication by the cube roots of unity,
$O_d=\{x\in H:x^3=d\}$, with
\[
 \prod_{x\in O_d}(1-xZ)=1-dZ^3,
 \qquad h_1(O_d)=h_2(O_d)=0,\quad h_3(O_d)=d.
\]
The distinct-element subset-sum recurrence, applied to $D$,
verifies that for each $3\leqslant j\leqslant13$,
\begin{equation*}
 \left\{\sum_{d\in J}d:J\in\binom Dj\right\}
    =\mathbb F_{169}.
\end{equation*}
For $r=3j+s$, $s\in\{0,1,2\}$, choose a fixed
$s$-subset $F_s\subseteq\{3,9\}$ and set
$T=F_s\cup\bigcup_{d\in J}O_d$.
These are disjoint subsets of $H\setminus\{1\}$.
As $J$ varies, $h_1(T)$ and $h_2(T)$ stay fixed, while
$h_3(T)$ varies by $\sum_{d\in J}d$ and hence takes every
value in the field.  Thus \eqref{eq:c3-anchor} holds for
the required $r=11,\ldots,29,39$.
The paired-point construction treats precisely $k=22,\ldots,59,78$.

For $q=256$, write
$\mathbb F_{256}=\mathbb F_2[u]/(u^8+u^4+u^3+u+1)$.
Let $V=\ker\operatorname{Tr}_{256/2}$,
$K=\ker\operatorname{Tr}_{256/4}\subsetneq V$, and fix
$t_0=u\in V\setminus K$; explicitly,
$\operatorname{Tr}_{256/2}(u)=0$ and
$\operatorname{Tr}_{256/4}(u)=1$.
The $63$ elements of $K^\times$ decompose into $21$
three-point orbits under multiplication by
$\mathbb F_4^\times$.  Writing $D=\{a^3:a\in K^\times\}$,
we have $\lvert D\rvert=21$ and, by the same exact recurrence,
\[
 \left\{\sum_{d\in J}d:J\in\binom Dj\right\}
 =\mathbb F_{256}\qquad(j=3,4).
\]
Both finite surjectivity assertions follow from the same
integer recurrence.  For $D=\{d_1,\ldots,d_N\}$, let
$M_i(j,\beta)$ count $j$-subsets of
$\{d_1,\ldots,d_i\}$ with sum $\beta$.  Then
\[
 M_i(j,\beta)=M_{i-1}(j,\beta)
                 +M_{i-1}(j-1,\beta-d_i),
 \qquad M_0(0,0)=1,
\]
with every other initial entry zero.  Exact arithmetic in the
specified field models gives
$\min_{\beta}M_{27}(j,\beta)\geqslant11$ for
$3\leqslant j\leqslant13$ over $\mathbb F_{169}$, and
$\min_{\beta}M_{21}(3,\beta)=1$,
$\min_{\beta}M_{21}(4,\beta)=12$ over
$\mathbb F_{256}$.  The program \nolinkurl{m67/verify_orbits.py}
recomputes these integer counts and constructs the required subsets.
The complement $V\setminus K$ is a union of $16$
additive $\mathbb F_4$-cosets.  Reserve $t_0+\mathbb F_4$,
which contains the forbidden point $t_0$, and use up to two
fixed points from its other three elements.  The remaining
$15$ cosets are available as fixed four-point cosets $Q$.
They satisfy
\[
 \prod_{x\in Q}(1-xZ)
   =1+Z^3+(a^4+a)Z^4,
 \quad h_1(Q)=h_2(Q)=0,\quad h_3(Q)=1,
 \qquad Q=a+\mathbb F_4.
\]
The third complete homogeneous symmetric function of such a coset
is $1$, not $0$.  Since these cosets are fixed, they add only a
constant to the third complete homogeneous symmetric function and
do not affect surjectivity.
For each required $r\ne8$, the certificate chooses
$j\in\{3,4\}$, $0\leqslant b\leqslant15$, and $0\leqslant s\leqslant2$
with $r=3j+4b+s$.  It constructs the disjoint union of
$j$ three-point orbits, $b$ fixed four-point cosets,
and $s$ fixed singleton points.  The resulting family lies in
$V\setminus\{0,t_0\}$ and has fixed values of the first two
complete homogeneous symmetric functions, while the third takes
every field value.  This gives
\eqref{eq:c3-anchor} for
\[
 r=11,\ldots,17,\quad 43,\ldots,55,
 \quad 56,58,60,62,
\]
covering the $42$ indicated pairs.  The remaining pair
$(256,17)$, where $r=8$, is addressed below.
\end{proof}

\begin{proposition}[Finite verification of the required subset condition]
\label{prop:seventh-finite-C3}
The conclusion for $m=7$ that the received words are not deep holes
holds for the remaining $15$ pairs with $q=243$, nine pairs with
$q\in\{64,81,121,125,128\}$, and $(q,k)=(256,17)$.
\end{proposition}
\begin{proof}
For $q=243$, the certificate provides eight families of
$5500$ stored $r$-subsets, one for each $53\leqslant r\leqslant60$,
inside the nonzero squares with $1$ removed.
Each family satisfies \eqref{eq:c3-anchor} for all $243^3$
triples.  These families treat $106\leqslant k\leqslant120$.
The nine further groups, specified by field and subset cardinality,
are listed below.
The last uses $6000$ stored subsets; the others use $2500$
each.
\begin{center}
\begin{tabular}{c c l}
$\,q\,$ & $r$ & required dimensions $k$ \\ \hline
$64$  & $4$   & $9$ \\
$81$  & $5,6$ & $11,12$ \\
$121$ & $8,9$ & $16,17,18$ \\
$125$ & $8,9$ & $17,18$ \\
$128$ & $6$   & $13$ \\
$256$ & $8$   & $17$
\end{tabular}
\end{center}
For each group the supplied program checks that
all points lie in the required parameter set, are distinct, and
avoid the specified excluded parameter,
recomputes $h_1,h_2,h_3$ and checks every $q^3$ triple in
\eqref{eq:c3-anchor}.  This establishes precisely the stated
$25$ pairs, without sampling coefficient triples.
\end{proof}

\begin{proposition}[Exhaustive coefficient verification over small fields]
\label{prop:seventh-small-fullcoef}
The conclusion for $m=7$ that the received words are not deep holes
holds for $q=16$, $3\leqslant k\leqslant6$, for $q=27$, $3\leqslant k\leqslant5$,
and for $(q,k)=(32,5)$.
\end{proposition}
\begin{proof}
Write $d=k+7$ and
$f=X^d+a_1X^{d-1}+\cdots+a_7X^{d-7}$.
The full-length code distance is invariant under
$g(X)=c^{-d}f(cX+b)$, $c\ne0$.
The first two coefficients transform by
\begin{align*}
 a_1'&=c^{-1}(a_1+db),\\
 a_2'&=c^{-2}\bigl(a_2+(d-1)a_1b+\tbinom d2b^2\bigr).
\end{align*}
Thus it suffices to check one representative from each orbit
of $(a_1,a_2)$ under affine variable changes.  The other
coefficients remain arbitrary, and the coefficient of $a_7$
in its transformation is $c^{-7}\ne0$.
For each of the eight pairs $(q,k)$, the ancillary files supply
a family $\mathcal F_{q,k}$ of $1200$ $(k+1)$-element sets.
The verifier checks, for every orbit representative and all
$(a_3,a_4,a_5,a_6)\in\mathbb F_q^4$, that
\begin{equation*}
 \left\{h_7(S)+a_1h_6(S)+\cdots+a_6h_1(S):
          S\in\mathcal F_{q,k}\right\}=\mathbb F_q.
\end{equation*}
Consequently every $a_7$ gives a vanishing remainder.
The numbers of orbit representatives are $3,2,3,2$ for
$k=3,4,5,6$ over $\mathbb F_{16}$; $3,3,4$ for
$k=3,4,5$ over $\mathbb F_{27}$; and $3$ for
$(q,k)=(32,5)$.  In total, $9{,}115{,}498$ normalized
choices of the four intermediate coefficients are checked,
and all the resulting $a_7$-values are covered.  The supplied
C++17 verifier performs this exhaustive computation with exact
finite-field arithmetic and checks the point-set cardinalities.
\end{proof}

\begin{proof}[Proof of Theorem~\ref{thm:seventh}]
By Lemma~\ref{lem:seventh-census}, only $375$ pairs need
separate treatment.  Proposition~\ref{prop:seventh-inherited}
covers $242+18+1=261$ of them.  The remaining
$114=39+42+25+8$ pairs are covered by
Propositions~\ref{prop:seventh-orbits},
\ref{prop:seventh-finite-C3}, and
\ref{prop:seventh-small-fullcoef}.
These five collections of parameter pairs are disjoint,
and their union is exactly the set of $375$ pairs in
Lemma~\ref{lem:seventh-census};
\nolinkurl{m67/verify_coverage.py} checks the
identity pair by pair.  Every stored subset consists of distinct
points in the required parameter set; the finite verifiers check
all stated affine targets or normalized coefficient classes.  The
lemmas constructing agreement sets and the remainder criterion therefore apply
to every received polynomial of degree $k+7$.
\end{proof}

\begin{proof}[Completion of the case $m=7$ in Theorem~\ref{thm:main}~(\ref{it:small-excess-classification})]
Apply Theorem~\ref{thm:seventh}.
\end{proof}

\begin{proof}[Proof of Theorem~\ref{thm:main}~(\ref{it:small-excess-classification})]
The assertion for $m=3$ follows from Theorems~\ref{thm:odd}
and~\ref{thm:even}.  The assertions for $m=4$ follow from
Theorems~\ref{thm:four-odd} and~\ref{thm:four-even}; the assertion
for $m=5$ follows from Theorem~\ref{thm:fifth}; and those for
$m=6,7$ follow from
Theorems~\ref{thm:sixth} and~\ref{thm:seventh}.
These results give exactly the two exceptional families in
\eqref{eq:segrefamily} and \eqref{eq:frobfamily} when
$m\in\{3,4,5,6,7\}$, with no restriction on $q$ beyond
$k+m\leqslant q-1$.
\end{proof}

\section{Consequences, limitations, and reproducibility}\label{sec:conclusion}

Theorem~\ref{thm:main} gives an effective general nonexistence threshold,
a complete classification above a quartic field-size threshold,
and an unconditional classification for $m=3,\ldots,7$.
For odd $m$ in this range there are no deep holes.  For
even $m$ in this range, the only exceptions are the characteristic-two
polynomials in dimension two listed in Theorem~\ref{thm:main}.  These statements concern prescribed interpolation degrees.

The Gao--Li estimate provides agreement sets in the middle dimensions;
value-set and secant arguments treat the small dimensions.  For a fixed small $m$, the remainder identities reduce the finite cases
to equations in the symmetric functions of subsets.  In the constructions for $m=7$ that adjoin a point, that point
is chosen so that the coefficient of $h_3$ is nonzero.  The finite
checks for these constructions use parameter sets that exclude the
parameter paired with the adjoined point.

The construction with pairs requires $h_4$ when $m=8$, so the current
certificates involving only $h_1,h_2,h_3$ do not by themselves settle
this case.  The present arguments do not classify arbitrary $m$.  For $q>19683$ and
$2\leqslant k\leqslant q-2$, Chen and Huang~\cite[Theorem~1.1]{ChenHuang2026}
classify all deep holes, with a parametrization by normalized
$o$-polynomials when $k=2$ in characteristic two.
Our classifications for $m=3,\ldots,7$ also apply to fields below that
threshold.  None of the present arguments asserts an analogue for
arbitrary punctured or generalized evaluation sets.

\subsection*{Computational reproducibility}
The finite parts of Theorem~\ref{thm:main}~(\ref{it:small-excess-classification}) are
computer-assisted proofs.  The cases $m=4$ and $m=5$ use the
certificates for normalized coefficient vectors and prescribed
symmetric-function values described above,
covering $95$ and $143$ residual parameter pairs, respectively.
For $m=6$, $265$ residual pairs are partitioned
as $147+2+103+3+10$.  The first four groups are covered by
exact finite-field certificates, and the last ten by the
dimension-two classification, including its finite check over
$\mathbb F_{16}$.  For $m=7$, the remaining
$375$ pairs are partitioned without overlap as
\[
 375=242+18+1+39+42+25+8.
\]
The first three groups use certificates for the equations involving
$h_1,h_2,h_3$, including certificates on subsets that omit one specified parameter at
$q=169$ and the direct calculation for $q=49$.  The next two use
three-point multiplicative orbits; the sixth uses $17$ further subset
families satisfying these equations, and the last uses eight
exhaustive enumerations of normalized coefficient vectors.

The ancillary directory \nolinkurl{anc/} contains the explicit
field models, point-set arrays, SHA256 manifests, and all
verification source code.  From that directory, the command
\texttt{python3 verify.py} runs the full finite verification by
default.  It requires Python~3.10 or later, NumPy, and a C++17 compiler.
The optional \texttt{--integrity-only} mode checks file identities
and does not check the mathematics.  No network access or
certificate-search program is needed.

The program \nolinkurl{core/verify_small_cases.py} checks the five
numerical rows in Lemma~\ref{lem:odd-middle-count} using exact
rational arithmetic, the two integer minima for $m=3$,
the coefficient table in the quintic Hermite obstruction, and
all $16^3$ degree-eight normal forms over $\mathbb F_{16}$.
The program \nolinkurl{core/verify_reductions.py} reproduces the
finite parameter reductions for $m=3,4,5$.
The component for $m=4,5$ checks $52$ subset families against
every required normalized coefficient vector.  Six further files
provide subsets with a fixed sum and every required value of the
second complete homogeneous symmetric function.  It recomputes
all required symmetric-function values and identities for adjoining
points, including the parity-dependent construction for $m=5$: when $k$
is odd, $0$ is accompanied by a distinct $z\in\{1,-1\}$ with
$z+a_1\ne0$.

The component for $m=6,7$ recomputes the two lists of remaining
parameter pairs and verifies that the parameter sets treated by
the certificates have exactly these unions.  Its $129$ shared subset families are stored and verified once
for every triple $(A,B,C)$ in \eqref{eq:c3-anchor}.  The same
component verifies the families with one excluded parameter, the
subset-sum assertions for three-point multiplicative orbits over
$\mathbb F_{169}$ and $\mathbb F_{256}$, the $17$ further subset
families, and the required enumerations of normalized coefficient
vectors.  Every required
triple $(A,B,C)$ or normalized coefficient vector is checked
exhaustively using the stored subsets, independently of
any randomness used to find them.  The analytic reductions and
normalization arguments are proved in the paper.  These
computations are reproducible computer-assisted certificates;
they are not a formal proof-assistant derivation.

\bigskip
\noindent\textbf{AI disclosure.} OpenAI's ChatGPT (GPT-5.6 Sol and GPT-6 Astra)
assisted with exploratory proof development, computational checks,
literature searches, and English-language polishing. The author takes full
responsibility for the correctness and content of this paper.

\providecommand{\bysame}{\leavevmode\hbox to3em{\hrulefill}\thinspace}
\providecommand{\MR}{\relax\ifhmode\unskip\space\fi MR }
\providecommand{\MRhref}[2]{%
  \href{http://www.ams.org/mathscinet-getitem?mr=#1}{#2}
}
\providecommand{\href}[2]{#2}

\end{document}